\documentclass[9.5pt, journal, compsoc]{IEEEtran}
\usepackage{xcolor}
\usepackage{colortbl}
\usepackage{diagbox}
\usepackage{makecell}
\usepackage{booktabs}%调整表格线与上下内容的间隔
\usepackage{multirow}
\usepackage{amsmath,amsfonts}
\usepackage{bm}
\usepackage{tikz}%箭头长度
\usepackage{cases}
\usepackage{amsthm}
\usepackage{upgreek}

\usepackage[linesnumbered,algoruled,lined]{algorithm2e}

\usepackage{verbatim}
\usepackage{textcomp}
\usepackage{comment}
\usepackage{subcaption}
\newtheorem{theorem}{Theorem}
\newtheorem{remark}{Remark}
\newtheorem{definition}{Definition}
\DeclareMathOperator*{\argmax}{argmax}

\begin{document}

\title{cRVR: A Cache-Friendly Approach to Enhancing Request Privacy for Online Video Services}

% \title{cRVR: A Novel Online Video System with  Privacy-Aware Video Requesting and Edge Caching}

\author{Xianzhi Zhang,
        Linchang Xiao,
        Yipeng Zhou,~\IEEEmembership{Member,~IEEE,}
        Di Wu,~\IEEEmembership{Senior Member,~IEEE,}        
        Miao Hu, ~\IEEEmembership{Member,~IEEE,}
        John C. S. Lui,~\IEEEmembership{Fellow,~IEEE,}
        and Quan Z. Sheng
        % <-this % stops a space
        
\IEEEcompsocitemizethanks{
\IEEEcompsocthanksitem Xianzhi Zhang, Linchang Xiao, Di Wu, and Miao Hu are with the School of Computer Science and Engineering, Sun Yat-sen University, Guangzhou, 510006, China, and the Key Laboratory of Machine Intelligence and Advanced Computing (Sun Yat-sen University), Ministry of Education, China. (E-mail: \{zhangxzh9, xiaolch3\}@mail2.sysu.edu.cn; \{humiao5,wudi27\}@mail.sysu.edu.cn.) (Di Wu is the corresponding author). 
\IEEEcompsocthanksitem Yipeng Zhou and Quan Z. Sheng are with the School of Computing, Faculty of Science and Engineering, Macquarie University, Australia, 2122. (E-mail: \{yipeng.zhou, michael.sheng\}@mq.edu.au). 
\IEEEcompsocthanksitem John C.S. Lui is with the Department of Computer Science \& Engineering, Chinese University of Hong Kong, China (E-mail: cslui@cse.cuhk.edu.hk).
}
}



\IEEEtitleabstractindextext{
\begin{abstract}
    As users conveniently stream their 
    %favoured 
    favorite online videos, video request records are automatically stored by video content providers, which have a high chance of privacy leakage. Unfortunately, most existing privacy-enhancing approaches are not applicable for protecting user privacy in video requests, because they cannot be easily altered or distorted by users and must be visible for content providers to stream correct videos. To preserve request privacy in online video services, it is possible to request additional videos that are irrelevant to users' interests so that content providers cannot precisely infer users' interest information. However, a naive redundant requesting approach %will 
    would 
    significantly degrade the performance of edge caches and increase bandwidth overhead.
    %accordingly. 
    In this paper, we are among the first to propose a \textit{Cache-Friendly Redundant Video Requesting (cRVR)} algorithm for User Devices (UDs) and its corresponding caching algorithm for the Edge Cache (EC), which can effectively mitigate the problem of request privacy leakage with minimal impact on the EC's performance. To tackle the problem, we first develop a Stackelberg game to analyze the dedicated interaction between UDs and EC, and obtain their optimal strategies to maximize their respective utility. For UDs, the utility function is a combination of both video playback utility and privacy protection utility.  We prove the existence and uniqueness of the equilibrium of the Stackelberg game. 
    %Furthermore, 
    Extensive experiments are conducted with real traces to demonstrate that cRVR can effectively protect video request privacy 
    by reducing up to 59.03\% of privacy disclosure 
    compared to baseline algorithms. Meanwhile, the caching performance of EC is only slightly affected.
\end{abstract}
}
\maketitle

\section{Introduction}\label{sec:introduction}

\IEEEPARstart{O}{nline} video streaming is one of the most profitable services for our daily life.
In the past decade, we have witnessed an unprecedented increase in online video and user population. According to~\cite{GlobalMediaInsight2024}, YouTube will provide 5 billion online videos for more than 122 million Internet users per day and the daily playback time will exceed a total of 1 billion hours in 2024.
In online video services, \textit{content providers (CPs)} can stream various videos such as movies, news, and TV episodes to billions of Internet users through \textit{edge cache (EC)}, \textcolor{black}{
%EC 
which can cache popular video contents on edge servers to reduce service latency~\cite{Tran2019, Qu2020, Yu2019, Han2022EdgeIntelligence, Yu2021}.} 

With the proliferation of online video markets, a rising concern is \emph{request privacy leakage}.
As users request videos from online content providers, content providers will automatically record request traces. 
It is challenging to prevent content providers from abusing user request traces. 
Besides, it is possible to further infer sensitive privacy such as gender, age, political views, and hobbies through analyzing users' request traces~\cite{Kang2022, Zhou2019a, Zhang2019a, Ma2017, He2022}. 
The leakage of request privacy may induce other anomalies, e.g., the spread of spam and scams~\cite{Ni2020, Xiao2018}. 
Therefore, a strategy to effectively protect user privacy is urgently needed to boost the security level of online video systems.
Nevertheless, user requests exposed in online video services cannot be easily protected, stemming from the following three challenges. 

{\color{black}
\textit{First}, users' request traces are automatically captured by 
%the 
a 
CP when users fetch a specific video content. While there are privacy protection methods, such as encryption-based methods in Information-Centric Networking (ICN)~\cite{Yuan2016, Xue2019}, these methods primarily focus on protecting the content information privacy of video providers or edge nodes. However, they cannot conceal the request privacy of users because user requests, i.e., indexes of movies,  are still visible to content providers and cannot be obfuscated by these methods.

\textit{Second}, user requests are generated based on users' playback preferences, which cannot be arbitrarily modified.
Therefore, some privacy protection methods distorting original information are ineffective. 
For example, differential privacy is widely applied to protect user privacy in recommender systems~\cite{Zhou2019a, Nguyen2016, Niu2020}, which assumes that users can arbitrarily alter their preferences to be exposed without affecting users' experience of recommendation.
However, if users randomly request videos, CPs and EC may miss many videos of users' genuine interest, resulting in disastrous user viewing experiences.
}

%\textit{And last but not least}, 
\textit{Third}, 
efficient video caching heavily relies on video population or user preferences inferred from video request records~\cite{zhang2022,Li2023}. Some naive methods, such as randomly fetching redundant videos to distort the actual records, may bring serious challenges to predicting video popularity, lowering the caching efficiency on EC.
In addition, the communication overhead between CPs and EC becomes substantial, and as such, spare video request generates additional communication traffic. 
This implies that user utility contradicts the utility of CPs and EC when redundant video requests are injected into online video service systems. Therefore, redundant video requests should be judiciously generated to mitigate the influence on edge video caching performance. 

{\color{black}
To address these challenges, we propose an innovative privacy-preserving redundant request algorithm that balances privacy and efficiency. 
Unlike current pre-fetching algorithms that solely consider caching performance, we introduce an information-theoretic utility function to guide UDs' strategies in generating more balanced requests that consider both privacy and caching efficiency. 
Furthermore, we develop a game-theoretic framework to facilitate the cooperation between the edge and user devices. Specifically, our \textit{Cache-Friendly Redundant Video Requesting (cRVR)} algorithm for UDs and the novel \textit{Edge Video Caching (EVC)} algorithm for the EC jointly protect privacy through redundant video requests without significantly compromising the caching performance. 
We analyze the interactions between UDs and the EC using the Stackelberg game to model the conflicting interests between UDs and the EC when redundant requests are injected. 
This dual approach ensures that the generated redundant requests can reconcile privacy protection and the overall efficiency of the edge video caching system.
%To sum up, 
In a nutshell, our main contributions 
%in this paper 
%can be 
are summarized as follows: 
\begin{itemize}
\item To the best of our knowledge, we are among the first to propose the cache-friendly redundant requesting algorithm to protect request privacy in online video streaming systems with a minimal impact on video caching performance. An information-theoretic utility function is designed to 
%insightfully 
guide UDs' strategies in generating video requests.
\item To analyze and minimize the impact of injected redundant requests, 
%our work is novel in modeling 
we model 
%the 
an 
interaction between UDs and the EC as a Stackelberg game. The optimal strategies for both parties are derived to optimize their respective utilities. 
We theoretically prove the existence and uniqueness of the equilibrium of the game.
\item We conduct extensive experiments by leveraging real-world video request traces collected from Tencent Video to verify that our algorithm can effectively diminish privacy disclosure compared to baseline algorithms. Meanwhile, 
%the 
cRVR only slightly compromises video caching performance on the EC. 
\end{itemize}
}

The remainder of this paper is organized as follows.  We 
%will 
first discuss the related works in Sec.~\ref{Sec:Related}. Then, the system model, including architecture, threat and privacy model, is described in Sec.~\ref{SYSTEM MODELS}. 
The utility functions of EC and UDs are elaborated in Sec.~\ref{Sec:utility function}.
Then, we formulate the problem to maximize utility and analyze the Stackelberg game in Sec.~\ref{PROBLEM FORMULATION}. The experimental results are presented in Sec.~\ref{Experiment} and finally, we conclude our paper in Sec.~\ref{sec:conclusion}.

\section{RELATED WORK}\label{Sec:Related}
%michael: remove the following sentence. Little value. 
%zxz：Okay, it would be better to remove it.
%In this section, we briefly review existing relevant works. 
\subsection{Privacy Leakage in Online Video Services}
Users mainly face three types of privacy threats in online video services.
Content providers (CPs) 
%give rise 
are the main contributors 
to the first type of privacy threat.
\textcolor{black}{Driven by commercial purposes, 
%content providers 
CPs 
have a strong incentive to gather and mine users’ private information~\cite{Ni2020, Nguyen2016}, which includes but is not limited to geographical location~\cite{Amini2011}, behaviour patterns~\cite{Zhang2019a,Han2024NOSSDAV}, personal information~\cite{Zhou2019a}.}
%, etc.} 
The above information can help to improve the CPs' service quality in terms of content caching~\cite{Shi2021} and recommendation~\cite{Niu2020, Guerraoui2017}.
%, and so on. 

The second type of privacy threat comes from the third-party edge cache (EC).
\textcolor{black}{Due to the latency issue, 
%content providers 
CPs commonly leverage the third-party edge cache to provide services, such as video caching~\cite{Han2024IOTJ}, streaming~\cite{Han2022JASC}, and data collection~\cite{Sanguanpuak2021, Xu2022}.}
%, etc. }
The edge cache owned or deployed by third parties~\cite{Xu2020a, Ni2020} also have incentives to obtain private data of users for various purposes (e.g., QoS enhancement~\cite{Xu2020a, Ma2017}, economic benefits~\cite{Li2016c}). 
The curiosity and accessibility of the EC are not within the control of content providers, 
%and this 
which 
brings additional risk for user privacy leakage~\cite{Xu2020a, Cui2020}.

Lastly, privacy can be infringed by user devices during video transmission~\cite{Yan2021}, distribution~\cite{Cui2020} and cache including D2D cache~\cite{Rui2022} and proactive client cache~\cite{Nikolaou2016}. 
Both user devices and edge devices can be leveraged to cache videos to improve users' 
quality of experience (QoE). It is widely assumed that UDs are trustworthy in existing works, though these devices may be compromised in reality~\cite{Wang2019b, Qiao2022}. 
It has been shown by previous works that attacks can be designed based on cached videos, such as cache side-channel attacks \cite{Acs2019, Sivaraman2021} and monitoring attacks~\cite{Nikolaou2016, Yan2021}. 
These attacks can obtain private data without permission, resulting in severe privacy leakage in online video services.

\subsection{Privacy Protection in Edge Caching Systems}
In general, there are mainly three 
%kinds 
groups of privacy protection methods in edge caching systems.
The {first} 
%kind 
group 
is \textit{cryptography-based methods}, e.g., encryption transmission~\cite{Araldo2018a,Yuan2016}, secured multi-party computing~\cite{Andreoletti2018} and blockchain-based methods~\cite{Qian2020}.
For example, Yuan \emph{et al.}~\cite{Yuan2016} proposed an encrypted video delivery protocol to prevent external attackers from accessing user privacy at in-network caching.
Andreoletti \emph{et al.}~\cite{Andreoletti2018} proposed an efficient edge caching scheme with secured multi-party computation to enhance the privacy protection of content providers and Internet service providers.    
Qian \emph{et al.}~\cite{Qian2020} proposed a blockchain-based edge caching architecture to protect user privacy. 
{\color{black}
%However, 
Unfortunately, 
these methods not only incur 
a 
heavy computational load~\cite{Ni2021}, but also are ineffective in preventing content providers from abusing personal private information~\cite{zhang2024}.

The {second} 
%kind is 
group of methods are \textit{disturbance-based}, such as differential privacy (DP)~\cite{Zhou2019a,Niu2020,Zhu2021,Zhang2022a}.
%For example, 
DP was introduced to protect personal interests in edge networks by~\cite{Zhou2019a, Nguyen2016, Niu2020}, which assumed that users can easily distort their exposed user-item interaction information. 
Zhou~\emph{et al}~\cite{Zhou2019a} proposed a privacy-preserving multimedia content retrieval system using DP and trust mechanisms to protect users' information while making recommendations and caching in an edge network. 
Zhu~\emph{et al.}~\cite{Zhu2021} added DP noises to user-item rating vectors before transmitting them to a global server when training the caching model, while Zhang~\emph{et al.}~\cite{Zhang2022a} proposed a DP-based caching method based on the private location transfer model and user preferences. 
%However, 
It is worth mentioning that DP was introduced to protect user privacy only in the model training process~\cite{Zhou2019a, Zhu2021, Zhang2022a} or assuming that users could easily distort their exposed information~\cite{Zhou2019a, Niu2020}, which is insufficient to preserve the request privacy in online video systems. 
Moreover, Zhang~\emph{et al.}~\cite{Zhang2023} and Nisha~\emph{et al.}~\cite{Nisha2022} preserved location privacy in edge caching systems by generating a $K$-anonymous set to distort the genuine request when a user sends a location-based query.

The {third} 
%type 
group of methods 
%is 
are \textit{distributed-based}, such as federated learning~\cite{Kang2022}. 
Despite their effectiveness in safeguarding record datasets~\cite{Qiao2022, Yu2021a, Yu2020b} during the caching algorithm training process, these methods fail to prevent content providers or other adversaries from collecting users' viewing traces when 
%users 
they 
retrieve video contents for watching or caching.}
%In light of 
To overcome the inefficacy of the 
existing methods, 
%above, 
we aim to develop a novel privacy protection method to enhance request privacy by generating redundant requests, 
%. In the meantime, we avoid 
%while avoiding 
without 
sacrificing the performance of edge caching. 
{\color{black}
Specifically, we aim to design a privacy-preserving strategy, which can be applied to various privacy-sensitive video streaming applications, e.g., short video recommendations~\cite{Zhou2019a, Guerraoui2017} and social network video transmissions~\cite{Xu2020a, Nikolaou2016, Wang2019b}, to enhance user privacy. This approach can be implemented by device manufacturers and operating system providers to provide system-level protection for user actions from being tracked by content providers.

\begin{table*}[!tbp]
\renewcommand\arraystretch{1.2}
\caption{Main notations used in the paper.}
\vspace{-4mm}
\begin{center}
\renewcommand{\arraystretch}{1.2}
\rowcolors{2}{white}{gray!25} 
\begin{tabular}{p{2.1cm}<{\centering} m{\linewidth-3.1cm}}
\toprule
Notation   & \multicolumn{1}{c}{Description}\\ 
\midrule
$i$ / $u$ / $t$ &  The index of any video content / UD / time slot, respectively. \\
 $\mathcal{I}$ / $\mathcal{U}$ / $\mathcal{T}$ &  The space of all video contents / UDs / time slots, respectively. \\
 $\pi(i) / \mathcal{K}_{\pi(i)}$ &  The category index of the video content $i$ and the set of video content  belonging to the category $\pi(i)$, respectively.\\
 $\bm{c}$  &  The normalized size vector of all video contents. \\
 $\bm{e}^t$ &The edge caching strategy of the EC at the time slot $t$.\\
 $\bm{x}_u^t$ / $\bm{y}_u^t$ / $\bm{a}_u^t$ & The genuine request vector/ redundant vector / public request vector of the UD $u$ at the time slot $t$, respectively.\\
 $\bm{w}_u^t$ / $\bm{r}_u^t$& The private / public request profile vector of the UD $u$ at the time slot $t$, respectively.\\ 
 $m_i^t$& The total number of UDs with public profile $r_{u,i}^t=1$ for the video content $i$ at the time slot $t$, i.e., $m_i^t = \sum_u r_{u,i}^t$.\\ 
 $P_{\text{Z}_i}$
 & The probability mass function of the disclosure probabilistic model, where $\text{Z}_i$ is a Bernoulli random variable. \\ 
 $\text{H}_u^t$ /  $H(r_{u,i}^t)$ & The privacy disclosure degree of the UD $u$ with the entire profile $\bm{r}_{u}^t$ / profile item $r_{u,i}^t$ at the time slot $t$, respectively.\\ 
 $\Delta n_i^t$& The number of UDs who first-time request the video content $i$ at the time slot $t$, i.e., $\Delta n_i^t = \sum_{\forall u: r_{u,i}^t =0}a_{u,i}^t$.\\
 $U_E$ / $U_u $ & The utility function for the EC / UD $u$, respectively. \\
 $L_{E,i}$ / $L_{u,i}$ & The profit function of QoE for the EC / UD $u$ to the video content $i$, respectively. \\
 $C_{E,i}$ / $C_{u,i}$ & The caching / requesting cost function for the EC / UD $u$ to the video content $i$, respectively. \\
$\epsilon_E$ / $\epsilon_u$& The unit caching cost for the EC / UD $u$, respectively.\\
 $D_{u,i} $ & The privacy disclosure function for the UD $u$ with respect to the video content $i$.  \\
 $d_{u,i}^t$ / $\hat{p}_{i}^t$ &  The view preference of the user $u$ for the video $i$. / The popularity of the video $i$ at time slot $t$.\\
 $\beta_E$ / $\beta$ / $\gamma$ &  The tunable parameters in utility function for EC's caching cost / UDs' requesting cost / UDs' privacy cost, respectively.\\
 $\bm{e}^{t*}$ / $\bm{a}_u^{t*}$ / $\bm{a}_{-u}^{t*}$ & The optimal EC's caching strategy. / The optimal public request strategy for the UD $u$ / other UDs except $u$, respectively.\\
$\hat{a}_{i}^{t}$ / $\Delta \hat{n}_i^t$ & The estimated value for $\sum_{\forall u}a_{u,i}^{t*}$ and $\Delta n_i^{t*}$, respectively, calculating bt the moving average value method. \\
 % $\mathbb{G}$ & The Stackelberg game for the edge-device collaborative caching problem.\\
\bottomrule
\end{tabular}
\end{center}
\label{Table:Major Notations}
\vspace{-4mm}
\end{table*} 
}

\begin{figure}[tpb]
\centering
\includegraphics[width=\linewidth]{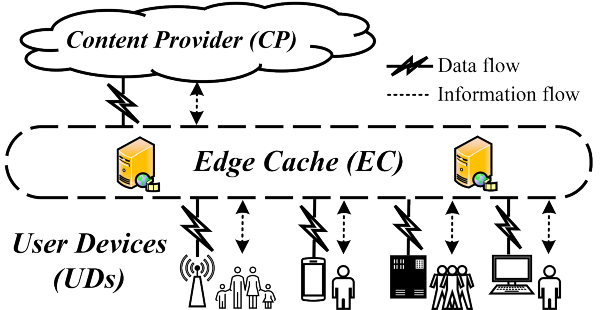}
\caption{The system architecture of a typical online video system.}
\label{Fig:EDGECACHING}
\vspace{-4mm}
\end{figure}

\section{System Models}\label{SYSTEM MODELS}
In this section, we introduce the system model, the user privacy model and the threat model.
%, orderly. 
{\color{black}
To begin with, we 
%make a brief introduction of 
first briefly introduce our framework architecture and the scenario studied by our work. 
Fig.~\ref{Fig:EDGECACHING} illustrates a typical architecture of online video systems~\cite{zhang2022}, which contains three major entities, namely, a \textit{Content Provider (CP)}, an \textit{Edge Cache (EC)}, and \textit{User Devices (UDs)}. Their main functionality and behavior are described in Sec.~\ref{sec:System Architecture}.
%michael: circled numbers not appear in any figure or later texts. They add little value, introduce disturbance in reading, and should be removed. 
%zxz：Okay, it would be better to remove it.
Generally, 
%\textcircled{1} 
when a user requests to watch a video, an EC or CP records the trace, 
%meanwhile providing 
and provide 
the required video streaming service to the UD.
This process, however, can lead to privacy leakage, as detailed in our risk model in Sec.~\ref{sec:Threat model}.
%\textcircled{2} 
To protect user privacy, we design an innovative redundant request strategy for user devices. 
The information-theoretic privacy model and user utility function are introduced in %Sections
Sec.~\ref{sec:privacy model} and~\ref{Sec: Utility of UDs}, respectively. 
%\textcircled{3} 
Simultaneously, to mitigate the overhead resource consumption and efficiency degradation by generating redundant requests, we develop an edge caching utility function for the EC to optimize edge-side efficiency, which is further discussed in Sec.~\ref{sec: Utility of EC}. 
%To facilitate the understanding of our analysis, 
For convenience, we summarize the frequently used notations 
%are summarized 
in Table \ref{Table:Major Notations}.
}

\subsection{System Architecture} \label{sec:System Architecture}

\subsubsection{\textbf{Content Provider (CP)}}\label{Sec:Content Provider}
%The 
A CP owns all videos denoted by $\mathcal{I}=\{1,\cdots,i,\cdots,|\mathcal{I}|\}$, which can be requested by UDs. These videos are organized into different categories such as movies, news, and sports. 
Let $\pi(i)$ represent the category of any video $i$ and $\mathcal{K}_{\pi(i)}\subseteq\mathcal{I}$ denote all videos in category $\pi(i)$.
Due to the large scale of modern online video systems, CP will cache videos by utilizing the edge cache, which is within the proximity of UDs to alleviate communication traffic load and to reduce transmission latency~\cite{zhang2022}.
Meanwhile, CPs will also proactively collect users' request traces to predict video popularity and users' interests~\cite{Ma2017, Zhou2019a}.

\subsubsection{\textbf{Edge Cache (EC)}}  
%The 
An 
EC can effectively offload the communication traffic of CPs. It makes caching decisions by taking video sizes and video popularity into account. 
Let $c_i\in(0,1]$ denote the normalized size of the video $i$, which is obtained by dividing the size of the video $i$ by the size of the largest video in $\mathcal{I}$. 
Without loss of generality, we consider a generic caching strategy of a particular EC entity (e.g., an edge server provides online video services in a certain area).
To ease our analysis, we suppose that the timeline is split into multiple time slots, i.e., $\mathcal{T}=\{1,\cdots,t, \cdots, T\}$.
Let the vector $\bm{e}^t= \left[e_{i}^t\right]^{|\mathcal{I}|},\, e_i^t\in [0,1],\forall i,t$ denote the caching strategy of the EC at time slot $t$, where $e_{i}^t$ is a real number indicating the ratio of the video $i$ cached in the EC. $e_{i}^t = 1$ means that the entire copy of the video $i$ is cached by the EC. 
The EC will update $\bm{e}^t$ periodically according to public information, e.g., request records, subject to the caching cost of the EC. 

Users in online video systems such as the one presented in Fig.~\ref{Fig:EDGECACHING} are susceptible to privacy invasion from at least two aspects: 
(1) \emph{Curiosity of CPs/EC}: CPs and EC have strong motivations to collect users' sensitive traces to enhance the quality of their services, such as recommendation~\cite{Guerraoui2017, Zhou2019a} and video caching~\cite{Zhou2019a, Shi2021}.
%, and so on. 
It is almost impossible to prevent CPs and EC from abusing collected users' traces~\cite{Xu2020a, Ni2020}. 
(2) \emph{Accessibility of EC}: To accommodate the huge user population, CPs widely leverage EC to provide video streaming services. 
Nonetheless, EC may be operated by third parties, such as Internet Service Providers (ISPs)~\cite{Xu2020a}, which can automatically obtain users' request traces. 
It is even 
more 
difficult for CPs to prevent third parties from abusing users' traces or malicious invasions launched by attackers to obtain users' traces~\cite{Araldo2018a, Ni2020}.

\subsubsection{\textbf{User Devices (UDs)}}  
The set of all users, (i.e., UDs), is denoted by $\mathcal{U}=\{1,\cdots,u,\cdots,|\mathcal{U}|\}$. UDs are always located at the edge of the network.   
UDs can obtain videos from 
%UDs' 
their 
local cache, the closest EC or the remote CP with a descending priority.
Normally, UDs request videos from the EC (or a CP) according to their view preferences, which unavoidably leaks user privacy~\cite{Ni2020, Guerraoui2017}. 
In our work, we significantly expand the video request strategy space of UDs by allowing UDs to proactively request and cache videos, including additional videos that are irrelevant to users' interests to conceal their viewing privacy. 
{\color{black}
When a user submits a genuine request $\bm{x}_{u}^t$ at time slot $t$, the UD $u$ will make redundant requests $\bm{y}_{u}^t$ to distort disclosed view preferences.  
Let $\bm{a}_{u}^t$ denote the public request vector sent by the UD $u$ at the time slot $t$, where $a_{u,i}^t = x_{u,i}^t |y_{u,i}^t$, and all $\bm{a}_{u}^t,\bm{x}_{u}^t, \bm{y}_{u}^t$ are $|\mathcal{I}|$ dimension vectors.}
Note that $\bm{x}_u^t$ is spontaneously generated by the UD $u$, and thus it is out of the scope of our strategic design. In contrast, $\bm{y}_u^t$ can be altered and our study focuses on how to generate $\bm{y}_u^t$ hereafter. 
{\color{black}
In our model, we assume that UDs are equipped with a certain small caching capacity that can be used to cache videos. Downloading videos in $\bm{y}_{u}^t$ can be regarded as the prefetching and proactively caching operation~\cite{Zhou2019a, Shi2021, Wang2019b} on the UD $u$ to further reduce network bandwidth consumption.
Due to the limited caching space, 
%the UD 
$u$ updates its cached videos according to the specific caching strategy\footnote{The caching algorithms in UDs are mature in the existing works, which are not the focus of this work.} (e.g., LRU~\cite{Quan2020}, utility-based methods~\cite{zhang2022, Wang2019b, Shi2021}) when its cache space is full.}

To better quantify the privacy leakage, we further define the public and private information for UDs. As all request traces can be recorded by the EC (or CPs), the request records will be the public information of UDs.
Let $\bm{r}_u^t=[r_{u,i}^t]^{|\mathcal{I}|}, r_{u,i}^t\in\{0,1\},\forall u,i,t$ denote the vector of public request profiles of the UD $u$ before time slot $t$. If $r_{u,i}^t=1$, it indicates that the UD $u$ has already requested $i$ before time slot $t$. Otherwise, the UD $u$ never requests the video $i$ before $t$.
The relation between $r_{u,i}^t$ and $a_{u,i}^t$ is
\begin{equation}
r_{u,i}^t = 1 - \prod\limits_{{\tau} < t}(1-a_{u,i}^{\tau}).
\end{equation}
The private profile is  $\bm{w}_u^t=[w_{u,i}^t]^{|\mathcal{I}|},w_{u,i}^t\in\{0,1\},\forall u,i,t$ for the UD $u$, which indicates whether 
%the UD 
$u$ has watched the video $i$ until $t$, i.e., 
\begin{equation}w_{u,i}^t = 1 - \prod\limits_{{\tau} < t}(1-x_{u,i }^{\tau}).
\end{equation}  
Note that the private profile $\bm{w}_u^t$ containing 
%the UD 
$u$'s privacy is sensitive and should be protected. The gap between $\bm{r}_u^t$ and  $\bm{w}_u^t$ includes redundant requests to conceal the privacy of $\bm{w}_u^t$.

\subsection{Threat Model}\label{sec:Threat model}
In our proposed framework, privacy threats in video requests primarily arise from the disclosure of users' video-watching patterns and preferences during interactions with %the 
an 
online video system. Historical video requests and fetching activities may inadvertently reveal sensitive information, enabling EC (utilized by CPs to enhance video streaming efficiency) to exploit insights into individual user preferences. Unauthorized access to such information poses a privacy threat, allowing CPs (or EC) to infer personal preferences and potentially breach user privacy.
The threat model settings mainly involve two distinct roles as follows: 
\begin{itemize}
    \item \textbf{Attackers, e.g., CPs (or EC)}: The primary threat comes from CPs, driven by the desire to improve Quality of Service (QoS) by deploying the EC infrastructure near users. Therefore, for simplicity, we consider CPs and the EC as the same risk entity in our threat model. CPs (or EC) are assumed to collect and analyze users' request traces, including detailed information about video-watching patterns. The attack scenario entails exploiting private information to gain insights into individual user behaviors and facilitating personalized recommendations or targeted advertising.
    \item \textbf{Defenders, e.g., UDs}: UDs play a crucial role in protecting user privacy by implementing privacy-preserving mechanisms, 
    %. These mechanisms 
    which 
    inject noises into the video pre-fetching process, challenging CPs (or EC) to extract sensitive information. UDs act as the first line of defence against privacy threats, ensuring confidential user request privacy while maintaining a positive QoE in online video systems.
\end{itemize}

\subsection{Privacy Model}\label{sec:privacy model}

We proceed to formally define privacy leaked by users based on their public profiles, i.e., request traces. 
A previous work~\cite{Guerraoui2017} has fully explored the amount of privacy disclosed given that a user clicks an item or not in recommender systems. 
We leverage this definition to model the privacy leakage of users in online video systems. 

{\color{black}
\begin{remark}
Intuitively,
%speaking, 
it is more difficult to infer a user's interest and discriminate one user from others if its public request profile is 
%very 
similar to others. 
In other words, a user profile with a significant discrepancy from others will expose more privacy 
risk 
such that the adversary can easily identify the user's interest.
\end{remark}}

Considering a popular video $i_0$, the probability of a user $u_0$ requesting and watching $i_0$ is higher than that of other cold videos~\cite{zhang2022}. The event $r_{u_0, i_0}^t=0$ will expose the privacy that the user likely does not prefer the video $i_0$. In contrast, if the user $u_1$ requests for a cold video $i_1$,  the event $r_{u_1, i_1}^t=1$ 
%implying 
implies 
that 
%the user 
$u_1$ likes the video $i_1$, 
and this makes it easy to distinguish the view preference of 
%the user 
$u_1$ from other users. 

Based on the above intuition, we adopt a notion of \textit{self-entropy}, defined in~\cite{Gray2011} and leveraged by~\cite{Guerraoui2017}, to measure the distance between a user profile and the public profile. Through self-entropy, we establish a disclosure probabilistic model to quantify the privacy of user profiles.
Let vector $\bm{m}^t=\sum_{\forall u}\bm{r}_{u}^t$ denote the total number of requests for each video until time slot $t$, which can be maintained independently by all UDs in the system given public profiles $\bm{r}_u^t,~\forall u$.
Based on the aggregated item $m_i^t$ in vector $\bm{m}^t$, a disclosure probabilistic model (i.e., a random variable $\text{Z}_i$) can be defined to evaluate disclosure privacy.
Here, $\text{Z}_i$ is a Bernoulli random variable (i.e., $\text{Z}_i \in\{0,1\}$), representing the event whether the video $i$ is requested or not.
The probability mass function of $\text{Z}_i$ is determined by:
\begin{equation}
\begin{split}
P_{\text{Z}_i}(r) =
\left\{\begin{aligned}
&\frac{|\mathcal{U}|-m_i^t}{|\mathcal{U}|},&
\begin{array}{l}
r=0;
\end{array} \\
&\frac{m_i^t}{|\mathcal{U}|},&
\begin{array}{l}
r=1.
\end{array}
\end{aligned}\right.
\end{split}
\label{EQ:disclosure probabilistic distribution}
\end{equation}

Given the disclosure probabilistic model, we now quantify the privacy disclosure degree of any public profile item $r_{u, i}^t,\forall i,\forall u,\forall t$ by calculating its value of self-entropy~\cite{Guerraoui2017, Gray2011} via Eq.~\eqref{EQ:H_ui}. A high \textit{information} value implies that there is a big gap with other users, 
%so it will lead 
leading 
to a higher privacy disclosure.
\begin{equation}
H(r_{u,i}^t)=-\log P_{\text{Z}_i}(r_{u,i}^t).
\label{EQ:H_ui}
\end{equation}

Similarly, we quantify the privacy disclosure for all videos. Let $\textbf{Z}=\left[\,\text{Z}_0,\text{Z}_1.\dots,\text{Z}_i,\dots,\text{Z}_{|\mathcal{I}|}\,\right]$ denote the probabilistic models for all videos, and all $\text{Z}_i$'s are independent.
The privacy disclosure of the user $u$ denoted by $\text{H}_{u}^t$ becomes
\begin{equation}
\begin{split}
\text{H}_{u}^t = \sum_i H(r_{u,i}^t) =  -\sum_{i}\log\left(P_{Z_i}(r_{u,i}^t) \right).
\label{EQ:H_u}
\end{split}
\end{equation}

Hence, we can address how a user should minimize privacy disclosure. The privacy disclosure $H(r_{u,i}^t)$ is large if a user requests an unpopular video $i$ (with a small $m_{i}^t$) because such request for unpopular videos can easily reveal a user's personal view preferences.
In contrast, the privacy disclosure is smaller if a user requests a popular video $i$ (with a large $m_{i}^t$). 
In the extreme case, if 
%the user 
$u$ only requests videos that have been requested by all other users, $\text{H}_{u}^t$ equals $0$, implying that there is no privacy disclosure because it is impossible to infer personal view preferences from its disclosed public profile.
Recall that $\bm{m}^t=\sum_{\forall u}\bm{r}_{u}^t$, we can also observe how users affect each other. $m_i^t$ is determined by requests of all users, which in turn affect the privacy disclosure of individual users. Therefore, the request decisions of each user can also change the privacy disclosure of others.

\section{Utility Functions in Privacy-aware Video Streaming Systems}\label{Sec:utility function}
 
At a high level, the objective of UDs is to conceal privacy and meets their playback requirements concerning the EC caching strategy. For the EC, the objective is to satisfy user requests at a low cost.  
Both UDs and the EC try to maximize their utility functions when making video requesting or caching decisions.

\subsection{\textbf{Utility of UDs}}\label{Sec: Utility of UDs}
UDs launch redundant requests to meet their playback requirement and distort disclosed view preferences. 
From  UDs' perspective, the utility consists of three parts: \textit{caching benefit}, \textit{privacy disclosure}, and \textit{caching cost}. Note that the utility of each UD is not only determined by the request strategy of the UD $u$ but also is dependent on other UDs' requests and the EC's caching decisions. Given the caching decision $\bm{e}^{t}$ of the EC and other UDs' strategies $\bm{a}_{-u}^{t}$, the utility function of 
%the UD 
$u$ can be defined as:
\begin{equation}
\begin{split}
&\mathrm{U}_u^t(\bm{e}^t,\bm{a}_u^t,\bm{a}_{-u}^t)=\sum_{\forall i: x_{u,i}^t =0}\mathrm{U}_{u,i}^t(e_i^t,a_{u,i}^t,\bm{a}_{-u,i}^t)=\\
&\sum_{\forall i: x_{u,i}^t =0} L_{u,i}^t(e_i^t,a_{u,i}^t)-\gamma\,D_{u,i}^t(a_{u,i}^t,\bm{a}_{-u,i}^t)-\beta\,C_{u,i}(a_{u,i}^t).%\\
%\text{s.t. }&\forall i\in \mathcal{I}:a_{u,i}^t\in\{0,1\}.
\label{Eq:U_u^t}
\end{split}
\end{equation}
Here, $a_{u,i}^t\in\{0,1\}$ and $L_u$ represent the caching benefit for 
%the UD 
$u$, $D_u$ is the privacy cost and $C_{u}$ is the caching cost. Note that $\gamma>0$ 
and $\beta>0$ are tunable parameters of privacy and caching cost for UDs, respectively. {\color{black} Users can set privacy parameters and cost parameters according to their 
%concern 
preferences. In this way, they can strike the balance between 
%their 
privacy and request costs according to their preferences and budget constraints.}

{
The utility definition in  Eq.~\eqref{Eq:U_u^t}  can be further explained as follows. 
The caching on UDs complements the caching on the EC. If the EC has cached the video $i$, it can provide high-quality video streaming for users, and thus repeatedly caching the video $i$ on UDs receives no benefit.  Otherwise, a UD can gain additional caching benefits if 
%the UD 
it 
caches a video that the EC has not cached.
According to the prior study~\cite{Wang2019b}, the benefit of proactively caching videos on UDs can be formulated as the product of the view preference, video popularity, and video size.

To summarize, the benefit for 
%the UD 
$u$ to request and cache the video $i$ with action $a_{u,i}^t$ concerning the EC's caching decision $e_i^t$ 
can be calculated by
\begin{equation}
\begin{split}
L_{u,i}^t\left(e_i^t,a_{u,i}^t\right) = a_{u,i}^t\ d_{u,i}^t\,\hat{p}_{i}^t\,c_{i}\,(1-e_i^t).
 \label{EQ:L_u,i}
\end{split}
\end{equation}
Here, $d_{u,i}^t$ in Eq.~\eqref{EQ:d_u,i^t} denotes the view preference of the user $u$ for the video $i$ and $\hat{p}_{i}^t$ in Eq. \eqref{EQ:hat p_u,i^t} is the popularity of the video $i$ at time slot $t$. The implication of Eq.~\eqref{EQ:L_u,i} is that the caching benefit can be brought by requesting and caching a popular video of the category falling in a user's interest, and the EC misses this video. $d_{u,i}^t$ and $\hat{p}_{i}^t$ will be further discussed as follows.}

According to~\cite{Wang2019b}, the view preference of the user $u$ for the video $i$, denoted by $d_{u,i}^t$, is relevant to historical viewing records of the same video category. 
Given the history of watched videos $\bm{w}_u^t$, 
$d_{u,i}^t$ is defined as 
%$d_{u,i}^t = \frac{\widetilde{w}_{u,i}^t}{t}\sum_{\forall i'\in \mathcal{K}_{\pi(i)}}\, w_{u,i'}^t$,  
\begin{equation}
 d_{u,i}^t = (1 - w_{u,i}^t)\cdot \frac{\sum_{\forall i'\in \mathcal{K}_{\pi(i)}} \, w_{u,i'}^t}{t}.
 \label{EQ:d_u,i^t}
\end{equation}
Here, $d_{u,i}^t$ is only valid if the user has not watched the video $i$~\cite{Wang2019b}. 
Recall that $\pi(i)$ represents the category ID of the video $i$. The term $\sum_{i'\in \mathcal{K}_{\pi(i)}}\, w_{u,i'}^{t}/t$ estimates the average view preference of $u$ for the video category $\pi(i)$ over $t$ time slots. Note that the computation of the user utility is based on $\bm{w}_u^t$, the private view history of the user $u$. Thus, users keep their utility functions as private information without disclosing them to others. 

A prior work~\cite{Shi2021} considered a time-varying dynamic model of video popularity, which can be used to model video popularity as: 
\begin{equation}
 \hat{p}_{i}^t = \sum_{{\tau}<t}\sum_{u'}a_{u',i}^{\tau}\exp[-\delta(t-{\tau})],
 \label{EQ:hat p_u,i^t}
\end{equation}
where $\delta$ is the decaying parameter.
Note that the actual viewing records are private information, and thus the popularity of a video $i$ can only be computed based on the public information $a_{u,i}^t$.

Since requesting videos has the risk of disclosing privacy, 
%therefore, 
users should factor privacy disclosure into their utility functions. 
Given the aggregated information $m_i^t$, we have defined the privacy disclosure by Eq.~\eqref{EQ:H_ui}. 
However, when users make request actions, privacy disclosure is a dynamic variable.   
Denote $\Delta n_i^t$ as the number of UDs who request the video $i$ at time slot $t$ for the  first time, i.e., $\Delta n_i^t = \sum_{\forall u': r_{u',i}^t =0}a_{u',i}^t
$, where $u'$ represents any UD who has never requested the video $i$ before time slot $t$. 
Thus, the relative change of privacy disclosure is:
\begin{equation}
\begin{split}
D_{u,i}^t&\left(a_{u,i}^t,\right.\left.\bm{a}_{-u,i}^t\right)=
r_{u,i}^t\log \frac{m_i^t }{m_i^t+\Delta n_i^t}  \\ 
&+\widetilde{r}_{u,i}^t\log \frac{ n_i^t}{ a_{u,i}^t(m_i^t +\Delta n_i^t)+\widetilde{a}_{u,i}^t(n_i^t-\Delta n_i^t)},
\label{EQ:D_u,i}
\end{split}
\end{equation}
where $\widetilde{r}_{u,i}^t = 1 - r_{u,i}^t$ and $\widetilde{a}_{u,i}^t = 1 - a_{u,i}^t$. {\color{black}
In addition, we denote $n_i^t=K-m_i^t$ as the popularity difference between the most popular video and any video $i$, where $K\geq\Delta n_i^t+m_i^t, \forall i$ is the number of users who have requested the most popular video in the whole online video system.
%michael: sentence below is not so easy to understand. Please improve. //it is clear now.
%zxz：Thanks for the suggestion. I have polished the sentence.
% Eq.~\eqref{EQ:D_u,i} shows the complexity of privacy disclosure in that user requests can affect privacy disclosure for each other. 
Eq.~\eqref{EQ:D_u,i} illustrates the complexity of privacy disclosure, demonstrating how individual user requests can impact the privacy exposure of others.} The implications of this interaction include:
\begin{itemize}
    \item \textbf{Case I, $r_{u,i}^t= 1$}: The privacy disclosure is reduced if more users request the same video because it becomes more difficult to identify the personal preferences of the video $i$. 
    \item \textbf{Case II, $r_{u,i}^t=0$ and $a_{u,i}^t=0$}: The privacy disclosure will be dilated if a popular video is not requested by 
    %the user 
    $u$, since it is easier to identify videos disliked by 
    the user. 
    %$u$. 
    \item \textbf{Case III, $r_{u,i}^t=0$ and $a_{u,i}^t=1$}: The change of privacy disclosure is determined by $\frac{m_{i}^t+\Delta n_i^t}{n_i^t}$. It implies that the request for the video $i$ dilates the privacy disclosure, such that identifying view preferences of 
    %the UD 
    $u$ becomes easier if $i$ is an unpopular video, and vice versa.
\end{itemize}

Additionally, requesting redundant videos also incurs additional traffic charges and bandwidth costs.
Therefore, we take into account the action cost on UDs 
%in 
as the following:
%manner:
\begin{equation}
\begin{split}
\label{EQ:C_u,i}
C_{u,i}\left(a_{u,i}^t\right)=a_{u,i}^t\,c_{i}\,\epsilon_u,
\end{split}
\end{equation}
where $\epsilon_u$ represents the total cost associated with the action $a_{u,i}^t$ of requesting and caching a unit size of the video $i$ for 
%the UD 
$u$, and $c_i\in[0,1]$ denotes the normal size of the video $i$.

\subsection{\textbf{Utility of EC}}\label{sec: Utility of EC}

From the EC's perspective, the gain 
%is from 
lies on 
accurately caching requested videos. The cost 
%lies in 
relates to the consumed caching capacity for video storage. 
Recall that $\bm{e}^t$ denotes the caching decisions of EC. Given request actions $\bm{a}^t$ of all UDs, the EC’s utility can be derived as follows:
\begin{equation}
\begin{split}
\label{Eq:U_E}
\mathrm{U}_{E}(\bm{e}^t, \bm{a}^t)=& \sum_{i}L_{E,i}(e_i^t, \bm{a}_i^t)-\beta_E\,C_{E,i}(e_i^t),\\
\text{s.t. }&\forall i\in \mathcal{I}:e_{i}^t\in[0,1],
\end{split}
\end{equation}
where $\beta_E$ is a positive tunable parameter to weigh the caching cost on the EC and $\bm{a}^{t}_i = [{a}^{t}_{u,i}]^{|\mathcal{U}|},\forall i,\forall t$.

The first term $L_{E,i}(e_i^t,\bm{a}_i^t)$ is the gain obtained by caching requested videos which should take the user's QoE into account. The EC is usually deployed by third parties such as ISPs and CDNs. The profit of the EC is directly related to the satisfaction of users served by the EC~\cite{Xu2020a, Xu2022}.  According to a prior study~\cite{Mao2017}, a user's QoE in online video services is a logarithmic function with respect to streaming bitrate.
% By extending the logarithmic satisfaction function with respect to streaming bitrate, the edge cache gain
Thus, we model $L_{E,i}(e_i^t,\bm{a}_i^t)$ as:
\begin{equation}
L_{E,i}\left(e_i^t,\bm{a}_i^t\right) = \sum_u a_{u,i}^t\log( 1+e_{i}^t\,c_i),
\end{equation}
where $c_{i}\in[0,1]$ is the normalized size of the video $i$ and $a_{u,i}^t\in\{0,1\}$ is the request action for 
%the video 
$i$ made by 
%the UD 
$u$. This definition implies that the EC receives positive gain only if the video $i$ is cached, i.e., $e_i^t >0$, and 
%the user 
$u$ requests 
%for 
the video $i$, i.e., $a_{u.i}^t=1$. 
The second term $C_{E,i}(e_i^t)$ is the caching cost of the EC.
Since the caching capacity of the EC is not free, we define the cost function $C_E(e_i^t)$ to cache the video $i$ as
\begin{equation}
C_{E,i}\left(e_{i}^t\right) =   e_{i}^t \, c_i\,\epsilon_E,
\label{EQ:C_{E}}
\end{equation}
where $\epsilon_E$ is the unit cost of edge caching. 

Our design objectives have two distinct goals. Firstly, our scheme should incentivize the EC (deployed by CPs at the edge network) to deliver high-quality caching services to UDs. Secondly, UDs should promptly determine an optimal redundant request strategy, striking a balance between users' QoE and privacy.
Given utility functions defined for UDs and the EC, our next step is to establish a game model in which UDs and the EC can 
%selfishly 
maximize their utilities. 

\section{Joint Edge Caching and Requesting Problem with Stackelberg Game}\label{PROBLEM FORMULATION}
In reality, the objective of UDs to conceal view preferences contradicts the objective of the EC that makes efforts to infer user preferences for making correct caching decisions.  
From the definition of utility functions in the last section, we also observe that the redundant video requests launched by UDs will affect the caching efficiency of the EC. 
Therefore, it is proper to model the interactions between the EC and UDs as a game in which the EC and UDs maximize their respective utilities based on the optimal actions of each other.
\textcolor{black}{In this work, we employ the Stackelberg game to deduce the best actions for the EC and UDs~\cite{Xu2020a,Han2024IOTJa}.}
In this section, a static scenario in a particular time slot will be analyzed first before we discuss the dynamic scenario over multiple time slots. 

\subsection{Problem Formulation}
Considering the roles of the EC and UDs\footnote{Without loss of generality, a CP and the EC can be regarded as an entity with the same target.}, we adopt a one-leader (i.e., the EC) and multiple-followers (i.e., UDs) Stackelberg game to model the system,  which is defined as:
%\begin{displaymath}
$
    \mathbb{G} := \{\left<E,\mathcal{U}\right>;\left<\bm{e},\bm{a}\right>;\left<\mathrm{U}_E,\mathrm{U}_{\forall u\in \mathcal{U}}\right>\}.
$
%\end{displaymath}
Here, $\bm{e}$ and $\bm{a}$ are the EC's and UDs' strategies, respectively. 
For a specific group of strategy $\left<\bm{e},\bm{a}\right>$, all participants can gain their payoff by utility functions, i.e., $\mathrm{U}_E$, or $\mathrm{U}_u$, correspondingly.

{\color{black}
Following the previous work in~\cite{Xu2020a}, we analyze the best responses for all participants in a static game scenario with complete information, where the EC and UDs have accurate knowledge of each other's utilities and public actions. 
%Although the EC is considered a risk entity in our threat model, 
%michael: hard to understand the sentence below. Please check. //ok
%the caching strategy for the rational EC involves being truthful. 
%zxz: Thank you for the suggestion. I have revised this paragraph to clarify the rationale behind the assumption.
This assumption is commonly applied in static game theory analysis~\cite{Xu2022,Xu2020a,Xu2019} and aligns with real-world requirements, as the EC must publish caching information to meet streaming service needs~\cite{Qian2020,Cui2020}. Additionally, verification methods~\cite{Xue2019,Cui2022,Xu2019} and attack techniques~\cite{Sivaraman2021,Acs2019} can be used to access this public information of UDs in edge caching environments.}
We will remove this assumption when introducing the algorithm design for a dynamic scenario.

In a static scenario of the game $\mathbb{G}$, the objectives of the EC and UDs are to maximize their utilities in each round, e.g., a time slot $t$, which can be achieved as follows. 
At the beginning of $t$, the EC makes an optimal caching strategy $\bm{e}^{t*}$ considering the best responses $\bm{a}^{t*}$'s of all UDs in this round.
Then, the UD $u$ acts as the follower and subsequently makes $\bm{a}_u^{t*}$ for video requests based on the optimal caching decision $\bm{e}^{t*}$ of the EC and other UDs' optimal strategies $\bm{a}_{-u}^{t*}$, independently and synchronously. 
In short, through static analysis of game $\mathbb{G}$, one can find the \textit{Stackelberg Equilibrium (SE)}, where none of the players have incentives to deviate from their strategies in each round of the game. 
It is equivalent to finding an optimal strategy group $\left<\bm{e}^{t*},\bm{a}_u^{t*}\right>$ at each time slot $t$, which guarantees that no participant can improve its own utility by unilaterally deviating from its current strategy, i.e.,
\begin{subequations}
\begin{align}
\textbf{Stage \uppercase\expandafter{\romannumeral1}: } &\mathrm{U}_{E}\left(\bm{e}^{t*}, \bm{a}^{t*}\right)  \geq \mathrm{U}_{E}\left(\bm{e}^t, \bm{a}^{t*}\right),\label{EQ:Stage I}\\
\textbf{Stage \uppercase\expandafter{\romannumeral2}: }&\mathrm{U}_{u}^t\left(\bm{e}^{t*}, \bm{a}_{u}^{t*},\bm{a}_{-u}^{t*}\right)  \geq \mathrm{U}_{u}^t\left(\bm{e}^{t*}, \bm{a}_{u}^t,\bm{a}_{-u}^{t*}\right),\label{EQ:Stage II}\forall u .
\end{align}\label{Eq:SE}
\end{subequations}
Based on the above two inequalities, we can deduce the equilibrium of the Stackelberg game, defined as:
%below:
\begin{definition}[Stackelberg Equilibrium]
\label{def:Stackelberg Equilibrium}
An optimal strategy group $\left<\bm{e}^{t*},\bm{a}_u^{t*}\right>$ constitutes a Stackelberg equilibrium of the joint caching game if inequalities in Eqs.~\eqref{Eq:SE} are satisfied.
\end{definition}

The next question is how to deduce the best caching decisions for each player, for which we can exploit the \textit{backward induction} method~\cite{Xu2020a,Xu2022}.

\subsection{The Static Game Analysis}
We first analyze the decision processes of each follower (i.e., each UD)  to derive the optimal strategy $\bm{a}_u^{t*}$ with any leader's strategy $\bm{e}^{t}$.
Then, we analyze the leader's (i.e., the EC's) optimal caching strategy $\bm{e}^{t*}$ based on UDs' best responses $\bm{a}_u^{t*}$'s to any $\bm{e}^{t}$. 

\subsubsection{\textbf{Analysis of Followers (UDs)}}\label{sec:Analysis of Followers}
Following the \textit{Backward Induction}
%~\cite{Xu2020a,Xu2022} 
method, we 
%irst elaborate the \textit{Nash Equilibrium (NE)} among all UDs, at which all followers (i.e., all UDs) in stage \uppercase\expandafter{\romannumeral2} can achieve their maximal utility and have no incentive to change their strategies
first study the best UDs' strategies
under the leader's (i.e., EC's) strategy $\bm{e}^t$. 
As we see from utility functions in Eqs.~\eqref{Eq:U_u^t}-\eqref{EQ:C_u,i}, all videos are independent of each other when computing utilities (i.e., $\frac{\partial^2 \mathrm{U}_{u}^{t}}{\partial a_{u,i}^{t}a_{u,j}^{t}}\equiv0\ , \forall i\neq j$). 
Therefore, we analyze the optimal caching strategy $\bm{a}_u^{t*}$ for each video $i$ independently. 
For simplicity, we relax the constraint $y^{t}_{u,i}\in\{0,1\}$ to $y^{t}_{u,i}\in[0,1]$ so as to make utility functions differentiable with respect to $y^{t}_{u,i}$ in our analysis. The final decision will be made by rounding off $y^{t}_{u,i}$.
We analyze the existence and uniqueness of the \textit{Nash Equilibrium (NE)} strategies for UDs at Stage \uppercase\expandafter{\romannumeral2} with given $e_i^t$.

\begin{theorem}[Existence]\label{theorem:eds exist}
In Stage \uppercase\expandafter{\romannumeral2}, any UD $u$ playing as a follower has the best strategy $a_{u,i}^{t*}$ on the video $i$ at the time slot $t$ with
\begin{subnumcases}{a_{u,i}^{t*}=\label{EQ:BestUDAct}}
1, &$x_{u,i}^t=1$; \label{EQ:BestUDActa}\\
\lfloor y_{u,i}^{t*} + 1/2 \rfloor\, &$x_{u,i}^t=0$,
\label{EQ:BestUDActb}
\end{subnumcases}
where $y_{u,i}^{t*}$ is:
\begin{subnumcases}{y_{u,i}^{t*}=\label{Eq:response function case}}
\max\{0,\min\{\Omega(\Delta n_i^{t*}),1\}\}, &$r_{u,i}^t=0$; \label{Eq:response function case a}\\
\argmax_{ y_{u,i}^{t}\in\{0,1\}}\mathrm{U}_{u,i}^t(e_i^t,y_{u,i}^{t},\bm{a}_{-u,i}^{t*}), &$r_{u,i}^t=1$,
\label{Eq:response function case b}
\end{subnumcases}
and 
\begin{equation}
\begin{split}
\Omega(\Delta n_i^{t*})=\frac{\gamma}{\beta\,c_i\,\epsilon_u-c_i\,d_{u,i}^t\,\hat{p}_{i}^t(1-e_i^t)}+\frac{n_i^{t}-\Delta n_i^{t*}}{N^{t*}_i}.
\label{Eq:response function B}
\end{split}
\end{equation}
Here $\Delta n_i^{t*} = \sum_{\forall u': r_{u',i}^t =0}a_{u',i}^{t*}$ and $N^{t*}_i= n_i^t - m_i^t- 2\Delta n_i^{t*}$. Additionally, $d_{u,i}^t$ is the view preference, $\hat{p}_{i}^t$ is the video popularity, $\epsilon_u$ is the unit caching cost, and $c_{i}$ is the video size.
\end{theorem}

Eq.~\eqref{EQ:BestUDActb} represents redundant requests distorting exposed privacy. 
Intuitively speaking, our proof is based on two cases. If $r_{u,i}^t=1$, it represents that the video $i$ has been requested by the UD $u$ before, and thus the utility to request it again for the UD $u$ is independent of the strategies of other UDs. In contrast, if $r_{u,i}^t=0$, 
%the UD 
$u$ needs to assess the utility of both requesting 
%or 
and 
not requesting this particular video $i$ separately to see which one is higher. 
We provide the detailed proof of Theorem~\ref{theorem:eds exist} in Appendix \ref{appendix:eds exist}. 

Next, we analyze the uniqueness of NE in Theorem \ref{theorem:eds unique}.
\begin{theorem}[Uniqueness]\label{theorem:eds unique}
Each UD has a unique optimal strategy $a_{u,i}^{t*}$ given the caching strategy of the EC (leader) $e_i^t$ and optimal strategies $\bm{a}_{-u,i}^{t*}$ of other UDs.
\end{theorem}

\begin{proof}
Given that Theorem \ref{theorem:eds exist} has guaranteed the existence of the NE, we prove that the NE is unique by proving that the best response functions $y_{u,i}^{t*}$ in Eq.~\eqref{Eq:response function case} are standard functions~\cite{Yates1995}, which are defined as follows:
\begin{definition}[Standard Function~\cite{Yates1995}]\label{definition:standard function}
A function $f(\bm{x})$ is standard if the following properties hold for $\forall\bm{x}\geq0$:
\begin{itemize}
\item Positivity: $f(\bm{x})\geq0$
\item Monotonicity: If $\bm{x}\geq\bm{x}'$, then $f(\bm{x})\geq f(\bm{x}')$
\item Scalability: $\forall k \geq 1,\ k\,f(\bm{x})\geq f(k\,\bm{x})$
\end{itemize}
\end{definition}
It is a sufficient condition for the uniqueness of the NE according to~\cite{Yates1995}.  
%According to 
Based on Theorem~\ref{theorem:eds exist}, the best response functions $y_{u,i}^{t*}$ in Eq.~\eqref{Eq:response function case b} when $r_{u,i}^t=1$ is naturally a standard function since the result of $\argmax$ operation is independent on the best strategies $\bm{a}_{-u,i}^{t*}$ of other UDs according to Eqs.~\eqref{Eq:U_u^t}-\eqref{EQ:C_u,i}.
In other words, given the value of $e_i^t$ and $y_{u,i}^{t}$, the function in Eq.~\eqref{Eq:response function case b} is a non-negative constant function corresponding to the best strategies $\bm{a}_{-u,i}^{t*}$ and guarantees the properties of the standard function, spontaneously. For the case $r_{u,i}^t=0$, we demonstrate that the best response function $y_{u,i}^{t*}$ in Eq.~\eqref{Eq:response function case a} satisfies the three properties of a standard function in Appendix \ref{appendix: eds unique}.
Therefore, UDs' best response functions defined in Eq.~\eqref{Eq:response function case} guarantee the uniqueness of the NE in Stage \uppercase\expandafter{\romannumeral2} in the game $\mathbb{G}$ with respect to any video $i$ and the proof of Theorem \ref{theorem:eds unique} is completed. 
\end{proof}

With the above analysis, the unique \textit{Nash Equilibrium} exists in Stage  \uppercase\expandafter{\romannumeral2} of game $\mathbb{G}$ among all UDs reacting to the EC's caching strategy $e_i^t$ on any video $i$. Recall that all strategies for different videos are independent. Hence, there exists a unique combined strategy $$\mathrm{U}_{u}^t\left(\bm{e}^{t}, \bm{a}_{u}^{t*},\bm{a}_{-u}^{t*}\right) \geq \mathrm{U}_{u}^t\left(\bm{e}^{t}, \bm{a}_{u}^t,\bm{a}_{-u}^{t*}\right),\forall t~,~\forall u,$$ where $\bm{e}^{t}$, $\bm{a}_u^{t*}$ and $\bm{a}_{-u}^{t*}$ are an arbitrary caching strategy of the EC, the best video requesting strategies of a specific UD $u$ and other UDs, respectively.

\subsubsection{\textbf{Analysis of Leader (EC)}}
Next, we analyze the optimal strategy for the EC in Stage \uppercase\expandafter{\romannumeral1} with the best response $\bm{a}^{t*}$ obtained in 
%michael: it is unclear. Can we clearly tell readers which section? i.e., Sec. xxx
%zxz: Okay, I have changed it.
%the section
Sec.~\ref{sec:Analysis of Followers}.
\begin{theorem}
\label{theorem:EC}
Given the optimal request  strategies $\bm{a}^{t*}$ of all UDs, the unique optimal strategy $e_{i}^{t*}$ on any video $i$ for EC is determined by:
\begin{subnumcases}{e_{i}^{t*}=\label{Eq:EC response function}}
0,&$\sum_{\forall u} a_{u,i}^{t*}< \uptheta$, \\
\frac{\sum_{\forall u}a_{u,i}^{t*}-\beta_E\,\epsilon_E}{\beta_E\,\epsilon_E\,c_i},&
 $\uptheta\leq\sum_{\forall u} a_{u,i}^{t*}\leq \Theta_i$,\\
1,&$\sum_{\forall u} a_{u,i}^{t*}> \Theta_i$,
\end{subnumcases}
where $\uptheta=\beta_E\,\epsilon_E$ and $\Theta_i=\beta_E\,\epsilon_E\,(1+c_i)$. 
Here, $\beta_E>0$ is the weight of the caching cost on the EC and $\epsilon_E$ is the unit cost for the EC caching defined in Eq.~\eqref{EQ:C_{E}}.
\end{theorem}

\begin{proof}
We prove Theorem \ref{theorem:EC} in Appendix \ref{appendix of theorem EC}.
\end{proof}

Thus, it can be proven that there exists a unique
$\mathrm{U}_{E}\left(\bm{e}^{t*}, \bm{a}^{t*}\right)  \geq \mathrm{U}_{E}\left(\bm{e}^t, \bm{a}^{t*}\right)$, where $\bm{e}^{t*}$ 
and $\bm{a}^{t*}$ are the optimal caching strategy on the EC and the optimal request strategies of all UDs, respectively.  
%By 
With above analysis, 
%above, 
we can finally conclude that the unique Stackelberg equilibrium exists in the static game $\mathbb{G}$ based on the optimal strategy group $\left<\bm{e}^{t*},\bm{a}^{t*}\right>$ following %the 
Theorem \ref{theorem:eds exist}-\ref{theorem:EC}.

\begin{algorithm}[!t]
\caption{EVC: Edge Video Caching Algorithm} 
\label{algorithm:Edge caching algorithm}
\KwIn{$\mathcal{I},\mathcal{U}$}
\KwOut{$\bm{e}^{*}$}
\For{$\forall t\in\mathcal{T}$}{
\For(\tcc*[f]{\textbf{STAGE \uppercase\expandafter{\romannumeral1}}}){$\forall i \in \mathcal{I}$\label{line: stage 1 begin}}{
%\tcp{\textbf{STAGE \uppercase\expandafter{\romannumeral1}}}\\
Collect public strategy $\bm{a}_{i}^{t-1*}$ from all UDs;\\
Evaluate $\hat{a}_{i}^{t}$ by $\bm{a}_{i}^{t-1*}$ based on Eq~\eqref{EQ:hat_a};\label{Line:EVC Evaluate}\\
Obtain the optimal caching strategy $e^{t*}_{i}$ by $\hat{a}_{i}^{t}$ based on Theorem \ref{theorem:EC} \label{Line:EVC strategy};
}
Announce the optimal caching strategy $\bm{e}^{t*}$.
}
\end{algorithm}

\begin{algorithm}[!t]
\caption{cRVR: cache-friendly redundant video requesting algorithm on any UD $u$ }
\label{algorithm:RVR algorithm}
\KwIn{$\mathcal{I}$, $\mathcal{U}$, $\forall t :\bm{x}_u$}
\KwOut{$\bm{a}_u^{*}$}
\For{$\forall t \in \mathcal{T}$}{
Get the public strategy $\bm{a}^{t-1*}$ from other UDs;\label{Line:update n begin}\\
\For{$\forall i \in \mathcal{I}$}{
Update $\Delta n^{t-1}_i$ by $\bm{a}^{t-1*}_u,\forall u$ and  UDs' public profile $\bm{r}^{t-1}_u,\forall u$;\\
Evaluate $\Delta \hat{n}_i^t$ by $\Delta n_i^{t-1}$ and $\Delta \hat{n}_i^{t-1}$ based on Eq.~\eqref{EQ:hat_deltan};\label{Line:update n end}\\}
Update $\bm{w}_u^t$ by the UD $u$ genuine request $\bm{x}_u^t$;\label{Line:cRVR optimal begin}\\
    \If(\tcc*[f]{\textbf{STAGE \uppercase\expandafter{\romannumeral2}}}){Local cache miss}{
        Get the public strategy $\bm{e}^{t*}$ from EC;\\
        
        \For{$\forall i \in \mathcal{I}$}{
            Update $d_{u,i}^t$, $\hat{p}_{i}^t$ by $\bm{w}_u^t$, $\bm{a}^{t-1*}_i$, respectively;\\
            Obtain the cRVR strategy $a_{u,i}^{t*}$ by $\Delta\hat{n}_i^t$, $d_{u,i}^t$, $\hat{p}_{i}^t$, $e^{t*}_i$ based on Theorem \ref{theorem:eds exist};\\
        }
        Fetch and cache videos from the EC or the CP with the redundant requests $\bm{a}_{u}^{t*}$;\\
        }
    \lElse{$\bm{a}_{u}^{t*}\leftarrow \bm{a}_{u}^{t-1*}$ and fetch the videos from local cache with $\bm{x}_u^t$ directly\label{Line:cRVR optimal end}}
    }
\end{algorithm}

\subsection{Algorithm Design in Dynamic Games}
{
We proceed to discuss how to design the \emph{edge video caching (called EVC)} algorithm for the EC and the \emph{cache-friendly redundant video requesting (called cRVR)} algorithm for UDs in a dynamic game $\mathbb{G}$ over multiple time slots.}  The details of EVC and cRVR are presented in Alg. \ref{algorithm:Edge caching algorithm} and Alg. \ref{algorithm:RVR algorithm}, respectively.

{The challenges for the design of EVC and cRVR algorithms can be illustrated from two aspects.} \textit{First}, the action information adopted by other players is incomplete to a particular player. For example,  $\bm{x}_u^t$, $\bm{d}_u^t$, and $\epsilon_u$ are the private information of 
%the UD 
$u$, which is invisible to the EC and other UDs, which however is useful for other UDs to make optimal actions.
\textit{Second}, in a large-scale online video system, user requests are generated at a high rate requiring that caching algorithms make decisions instantly so that user requests can be satisfied in time. 
Given these two challenges, the design principle of our algorithms is to make caching decisions with the best estimation of incomplete information in an efficient way.

In each slot, the EC as the leader makes video caching decisions first. According to Theorem~\ref{theorem:EC}, the EC needs the knowledge of $a_{u,i}^{t*}, \forall u$ in Eq.~\eqref{EQ:BestUDAct}-\eqref{Eq:response function B} before making the best caching decisions.
However, such information is unavailable before UDs as followers take their actions. 
A feasible solution is to estimate $\sum_{\forall u}a_{u,i}^{t*}$ based on historical records shown in Line \ref{Line:EVC Evaluate} (Alg.~\ref{algorithm:Edge caching algorithm}).
A straightforward method is to estimate $\sum_{\forall u}a_{u,i}^{t*}$ using   $\sum_{\forall u}a_{u,i}^{(t-1)*}$. 
To make our estimation robust and resilient to noises, we propose to estimate $\sum_{\forall u}a_{u,i}^{t*}$ using the moving average value (MAV) of historical records, which is expressed as follows: 
\begin{equation}
     \hat{a}_{i}^{t} =   (1-\rho)\sum_{u}a_{u,i}^{t-1*}+\rho\,\hat{a}_{i}^{t-1}.
    \label{EQ:hat_a}
\end{equation}
%\noindent 
Here, $\rho$ is a hyper-parameter to tune the weight of historical records at each time slot and the initial value $\hat{a}_{i}^{1} = 0$. 
%michael: black or black? please check. ok
%zxz: It should be black. These changes were made based on the review comments from the previous conference.
The computation complexity to evaluate $\hat{a}_{i}^{t}$ is $O(|\mathcal{U}|)$, where $|\mathcal{U}|$ is the total number of users served by the EC. With the evaluated $\hat{a}_{i}^{t}$, the EC can determine 
%the 
an 
optimal caching strategy for each video with $O(1)$ computation complexity following Theorem \ref{theorem:EC}.
The overall time complexity is $O(|\mathcal{U}|\cdot|\mathcal{I}|)$ for making edge caching decisions of all videos, where $\mathcal{|I|}$ is the video population.

%michael: black or black? Pls check. //ok
%zxz: It should be black. These changes were made based on the review comments from the previous conference.
UDs as the followers can make their best responses based on $\bm{e}_t^*$ announced by the EC.
For individual UDs to make request decisions, it is difficult to obtain the value of $\Delta n_i^t = \sum_{\forall u: r_{u,i}^t =0}a_{u,i}^t,\forall i$,  which is dependent on the decisions of all UDs. 
Similarly, this number can be well estimated from historical records by: 
\begin{equation}
    \Delta \hat{n}_i^t =  (1-\rho)\Delta n_i^{t-1}+\rho\Delta \hat{n}_{i}^{t-1},
    \label{EQ:hat_deltan}
\end{equation} 
%\noindent 
where the initial value $\Delta \hat{n}_{i}^{1} = 0$. The details of estimation are described in Line \ref{Line:update n begin}-\ref{Line:update n end} (Alg. \ref{algorithm:RVR algorithm}). 
The computation complexity is $O(|\mathcal{U}|)$ when evaluating  $\Delta \hat{n}_i^{t}$ for the video $i$ with UDs' public request actions $\bm{a}_u^{t-1*}, \forall u$, public profile $\bm{r}^{t-1}_u, \forall u$, and the estimated value $\Delta \hat{n}_{i}^{t-1}$ in last time slot $t-1$. 
When a new genuine request $\bm{x}_u^t$ arrives, 
%the UD 
$u$ can make the optimal requesting strategies of each video in $O(1)$ time complexity following Theorem \ref{theorem:eds exist} with the evaluated value $\Delta \hat{n}_i^{t}$. The overall time complexity is $O(|\mathcal{I}|)$ for making redundant requests of all videos, where $\mathcal{|I|}$ is the video population. The detailed process to generate redundant requests is described by Line \ref{Line:cRVR optimal begin}-\ref{Line:cRVR optimal end} (Alg.~\ref{algorithm:RVR algorithm}).

\vspace{-2mm}
\section{Experiments}\label{Experiment}

In this section, we compare the proposed cRVR with other algorithms to demonstrate that cRVR can significantly reduce privacy disclosure without affecting the caching performance on the EC. 
We set up an edge cache that is equipped with a certain caching capacity to provide video streaming services for UDs. UDs' requests that cannot be served by the EC will be directed to the content provider. 

\subsection{Datasets and Experimental Settings}

\subsubsection{\textbf{Dataset Settings}}
We mainly 
%michael: mixed tense usage. Stick to present tense. 
%Okay, thanks for the suggestion.
%conducted 
conduct 
experiments on two datasets to verify the performance of our approach:
% \begin{itemize}
    % \item 
(1) \textbf{\textit{Tencent (TC) Dataset}.}
Due to the large size of Tencent Video\footnote{https://v.qq.com/} Dataset in~\cite{zhang2022}, we only sample requests from $1,000$ UDs in a particular city over 30 days. There are $20,999$ unique videos with more than $222K$ requests in total.
Each request record in our dataset is represented by the metadata $<$UserID, VideoID, TIME, VideoType, Size$>$. 
We split the entire dataset into two subsets based on the date.  The first subset is used to initialize the online video system containing requests in the first 12 days. The second subset contains requests for the rest of 18 days, which is used for testing. 
% \item 
(2) \textbf{\textit{MovieLens Dataset},} which is generated with the MovieLens 25M (ML-25M) Dataset~\cite{movieLens}, 
%which is 
a public video rating dataset. We use the user-video rating time as the time when the video is played by 
%the 
a 
user. Due to the extreme sparsity of the rating time of the MovieLens Dataset, the time span of the MovieLens Dataset is scaled down to 30 days. We also sample 1,000 users from ML-25M  for experiments to test the performance of our algorithm. There are a total of $158K$ viewing records in the sampled dataset related to 10,002 videos. 
This dataset is also split into two subsets based on the date, the same as the first dataset.
\textcolor{black}{We 
%also 
provide 
the 
detailed information 
%about 
of 
these two datasets in Appendix~\ref{appendix: dataset}.}
% \end{itemize}

\subsubsection{\textbf{Simulation Settings}}
{\color{black}
In our experiments, all strategies are programmed in Python and tested on an AMD server with four EPYC™ 75F3 @ 2.95 GHz processors with 502 GB memory. 
%In these experiments, 
The EC service is implemented in the host process, while each UD strategy is executed in an independent sub-process. Communications between the host process and the sub-processes are handled via a socket protocol.
To simulate the limited capacity of personal devices, the cache capacity for UDs is set to 0.5\% of the total number of videos in each dataset. A utility-based caching strategy~\cite{zhang2022} is used as the default eviction method for video replacement when the local cache of UDs is full.}
Each time slot in both 
%two 
datasets is set as $10$ minutes. The hyper-parameter $\rho$ of the MAV method in the dynamic game is set as $0.9$.
Besides, both the caching cost $\epsilon_E$ for the EC and $\epsilon_u$ for all UDs are fixed as $1$ for convenience.
We set the decay parameter $\delta = 0.01$ in $\hat{p}_i^t$ according to~\cite{Shi2021}. 
Experiments are also conducted to study the influence of the parameters $\gamma$, $\beta$, and $\beta_{E}$ in the utility function.

\subsection{\textbf{Baselines}} 

To evaluate the performance, we compare different combinations of video requesting 
%algorithms 
and caching algorithms. For the video requesting algorithms, we compare our approach with the following:
%alternatives, namely:

\begin{itemize} 
\item \emph{\textbf{NR} (No Redundant Requesting)}, which does not request any additional video and is used as the benchmark.
\item \emph{\textbf{RANDOM} (Random Redundant Requesting)}, which randomly requests redundant videos to protect privacy. 
%\textbf{GIAS}
%michael: should it be AIGS? pls check.//ok
%zxz: Apologies for the mistake. It should be the Generating Initial Anonymous Set (GIAS) algorithm. I have also rephrased other expressions.
{\color{black}
\item \emph{\textbf{GIAS} (Generating Initial Anonymous Set)}~\cite{Zhang2023}, which preserves location privacy in edge computing by generating an initial $K$-anonymous set when a user sends a location-based query. We have modified the GIAS algorithm to handle client requests for video services by setting $K$ the same as the average number of requests in cRVR under different testing conditions. 
%michael: what do we mean by "this application"? do we mean this research?//ok
%zxz: I think it should indicate GIAS algorithm and I have revised the following sentences.
Most of the parameter settings for GIAS are consistent with those provided in the original research paper. For the key weight coefficients $\alpha$ and $\beta$
, which are not explicitly defined in the article, we set 
$\alpha = 0.5$ and $\beta = 0.5$ based on the constraint $\alpha + \beta = 1$ in~\cite{Zhang2023}.
% The most settings parameters used in  this application 
% are consistent with those in the original research article. For other undefined parameters, we follow the research paper,  
% %article's guidance, 
% setting $\alpha = 0.5$ and $\beta = 0.5$ with $\alpha + \beta = 1$.

\item \emph{\textbf{PPVR} (Point Process-based Video Requesting)}~\cite{Shi2021}
%: This which 
is an innovative model based on point process model that predicts future video request rates and adapts to various device types. We implement PPVR on the UD side as a video request method, selecting videos with the highest predicted request rates as redundant videos.
%michael: check above. sentence are not right. //ok
%zxz: thank you for the advice. I have rewritten the above sentences.
All model settings remain consistent with 
those in 
the original work \cite{Shi2021}.}
\end{itemize} 

To keep a fair comparison, we fix the average number of requesting videos for RANDOM, GIAS, and PPVR, which is identical to that of cRVR in each experiment.

%In terms of 
For the caching algorithms on the EC, we consider the following three baselines, which run with the same caching capacity as that of EVC in different cases:

\begin{itemize} 
\item \emph{\textbf{LRU} (Least Recently Used)}, which replaces the least recently used video with a newly requested video.  
\item \emph{\textbf{LFU} (Least Frequently Used)}, which replaces the least frequently used video with a newly requested video. 
% \item \emph{\textbf{PPVC} (Point Process-based Video Caching)}, which is a state-of-the-art learning-based caching algorithm~\cite{Shi2021}. 
\item \emph{\textbf{PPVC} (Point Process-based Video Caching)}~\cite{Shi2021}, 
%: This 
which is a learning-based caching algorithm 
that 
predicts future video request rates and caches videos with the highest predicted request rates. We have reused this method to optimize edge caching, maintaining the same settings as those in the city server case in the original work~\cite{Shi2021}.
\end{itemize}

\subsection{\textbf{Metrics}} 
The following metrics are used to evaluate the performance of different algorithms:
\begin{itemize}
\item {\color{black} \emph{\textbf{CHR} (Cache Hit Ratio)},  which is defined as the number of watched video hits at all UDs divided by the total number of watched videos over the entire test period. CHR is the most widely used metric for evaluating caching performance.}

\item \emph{\textbf{PDR} (Privacy Disclosure Ratio)}, which is computed by $$PDR = \frac{\sum_i H(r_{u,i}^T)}{\sum_i H(w_{u,i}^T)}$$ according to Eq.~\eqref{EQ:H_ui}, where $T$ is the last time slot in the dataset. 
A lower value of \emph{PDR} implies less privacy exposure of a UD.
% {\color{black}
% \item \emph{\textbf{RHR} (Recommendation Hit Rate) \textbf{Degradation}}, which calculates the averaged degradation of RHR among all users when using a recommendation algorithm to recommend videos for users based on their original profiles and profiles exposed by EDs. 
% A larger degradation of RHR implies a stronger privacy protection. We implement the recommendation algorithm proposed in~\cite{He2017} for our experiments with the same settings and RHR metric. }
{\color{black}
\item \emph{\textbf{BCR} (Bandwidth Consumption Ratio)}, which is calculated by dividing the actual bandwidth consumption of the request algorithm by the bandwidth consumption of NR, not requesting any additional video. A higher BCR implies more bandwidth consumption by the request algorithm. To assess the impact of the request algorithm on both UDs and the core network on the CP, we test BCR separately for each side.
}

\item \emph{\textbf{BOR} (Bandwidth Offloading Ratio)}, which is defined as the fraction of bandwidth served by the EC for all requests. BOR can be used to measure the caching efficiency of the EC. 
\end{itemize}

\begin{figure*}[tbp]
\includegraphics[width=0.92\linewidth]{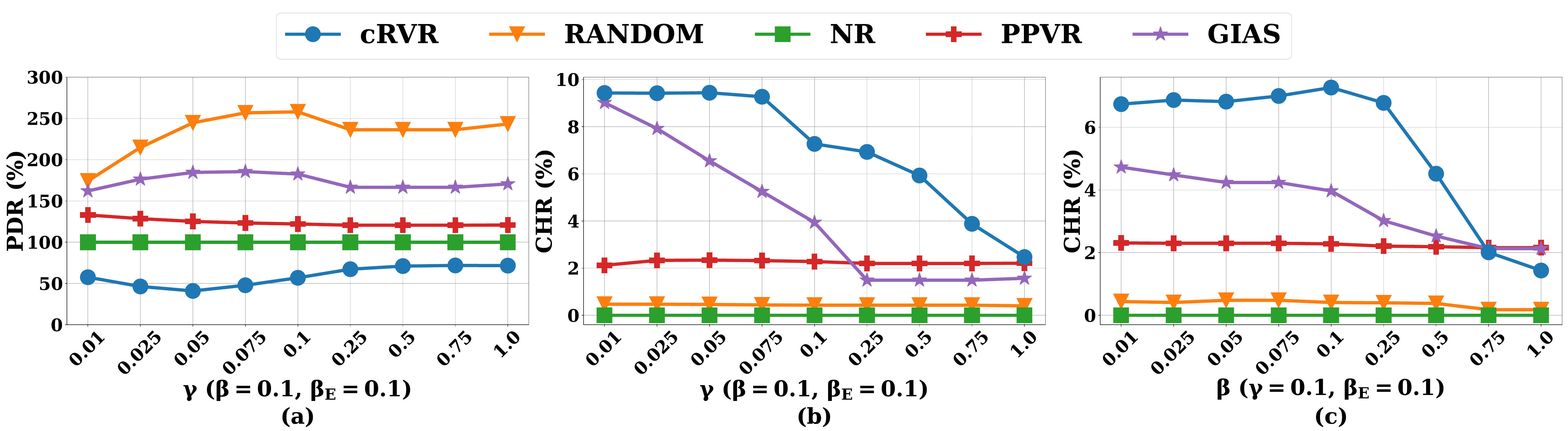}
\vspace{-2mm}
\caption{\color{black}Comparing privacy and caching performance at UDs tested on the TC Dataset with EVC running at the edge. A lower PDR is preferred to protect private information and a higher CHR at the UDs results in a better viewing QoE.}
\vspace{-4mm}
\label{Fig:tc_UD}
\end{figure*}

\begin{figure*}
\includegraphics[width=0.92\linewidth]{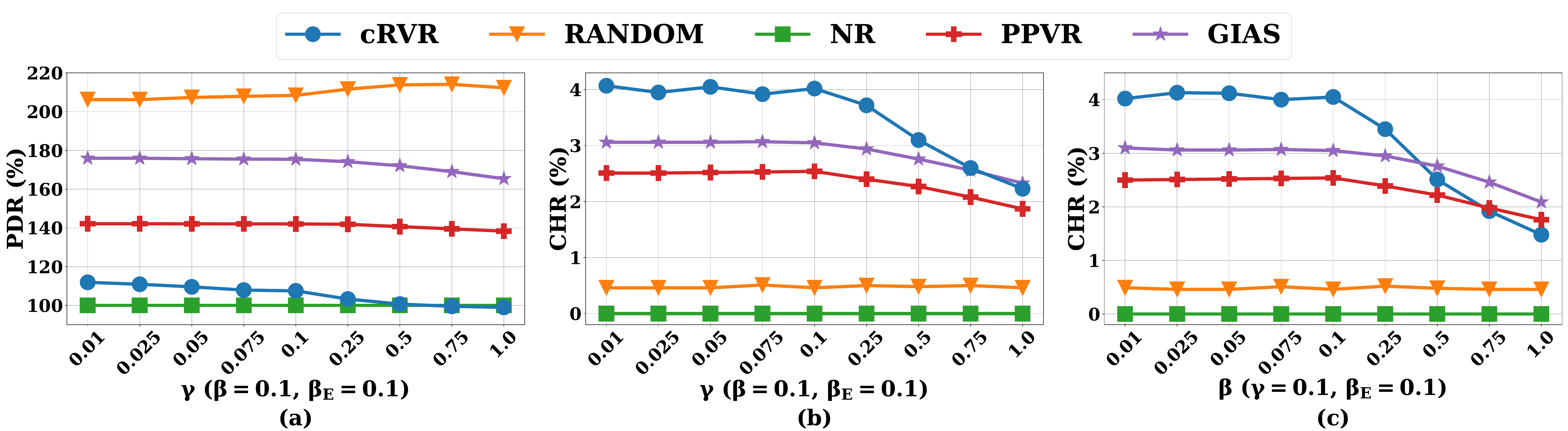}
\vspace{-2mm}
\caption{\color{black}Comparing privacy and caching performance at UDs tested on the ML-25M Dataset with EVC running at the edge. A lower PDR is preferred to protect private information and a higher CHR at the UDs results in a better viewing QoE.}
\label{Fig:UD_ml25m}
\vspace{-4mm}
\end{figure*}

\vspace{-4mm}
\subsection{Experimental Results}

We implement our EVC algorithm to determine videos cached on the EC based on requests from UDs. 
Then, the cRVR, PPVR, GIAS, RANDOM, and NR algorithms are implemented on UDs separately to compare their performance in terms of privacy protection and caching efficiency. 

\vspace{-2mm}
\subsubsection{\textbf{
%Evaluating 
Privacy Protection Performance on UDs}} 
We measure the value of PDR for each video requesting algorithm on UDs by tuning the privacy weight $\gamma$ from $0.01$ to $1.0$.
Through comparing cRVR with baselines in Fig.~\ref{Fig:tc_UD}(a) using the 
TC dataset, we draw the following 
%insightful 
observations:
\begin{itemize}
    \item cRVR achieves the lowest PDR for all experiment cases by reducing PDR 
    %by 
    \textcolor{black}{59.03\%} at most and \textcolor{black}{40.94\%} on average against the NR strategy. 
    \item PPVR requests videos based on UDs' interest and thereby the PDR of PPVR is higher than NR.
    \item RANDOM is the worst one, implying that blindly generating redundant requests cannot conceal privacy.
    \item cRVR is more conservative as the weight of the privacy disclosure term $\gamma$ in the utility function gets bigger. In other words, UDs' utility will be dominated by privacy disclosure with a large $\gamma$, and thus cRVR will only generate requests due to privacy benefits. 
\end{itemize}

%michael: black or black?
%zxz: It should be black. These changes were made based on the review comments from the previous conference.
We conduct additional simulations using the ML-25M dataset under the same settings. The comparison of the PDR is depicted in Fig.~\ref{Fig:UD_ml25m}(a), revealing a distinct pattern from the one illustrated in Fig.~\ref{Fig:tc_UD}(a). Specifically, the PDR of cRVR exhibits a decrease correlated with the increment of the privacy weight parameter and the second-best PDR is observed when $\gamma$ is small.
This phenomenon arises due to the sparsity records of ML-25M~\cite{movieLens}, leading to a scarcity of popular video in the online video streaming system. Caching videos driven by the utility based on personal preferences poses substantial privacy risks for all algorithms, except NR, which exclusively requests videos actively watched by users.
Nevertheless, as 
%the 
$\gamma$ increases, for instance, when $\gamma=0.75$ and $1$, our approach generates requests based on privacy benefits, resulting in lowest PDR values (e.g., \textcolor{black}{$\text{PDR} = 99.50\%$} and \textcolor{black}{$98.96\%$}) compared to 
%the 
all other baselines. Remarkably, as shown in Fig.~\ref{Fig:tc_UD}(a), cRVR also demonstrates enhanced performance with higher request arrival rates. These findings collectively demonstrate the adaptability of our cRVR across diverse online video systems with different viewing patterns.

\begin{table*}[ht]
\centering
\caption{\color{black}BCR performance in the UD side and the CP side while varying UDs' cost parameter $\beta$ under the TC dataset.  A lower BCR is preferred in resource consumption.}
\label{table:consumption in TC}
\renewcommand{\arraystretch}{1.2}
\resizebox{0.8\linewidth}{!}{
\begin{tabular}{c|c|>{\centering\arraybackslash}p{1cm}|>{\centering\arraybackslash}p{1cm}|>{\centering\arraybackslash}p{1cm}|>{\centering\arraybackslash}p{1cm}|>{\centering\arraybackslash}p{1cm}|>{\centering\arraybackslash}p{1cm}|>{\centering\arraybackslash}p{1cm}|>{\centering\arraybackslash}p{1cm}|>{\centering\arraybackslash}p{1cm}}
\toprule
\multirow{2}{*}{BCR} &\multirow{2}{*}{Methods} & \multicolumn{9}{c}{UDs' Requesting and Caching Cost $\beta$ } \\ 
\cline{3-11}
&&0.01  &   0.025 & 0.05 &  0.075 & 0.1 & 0.25 & 0.5  &   0.75& 1.0 \\ 
\midrule
\multirow{4}{*}{\makecell{UD Side}} &cRVR & \textbf{14.60}&\textbf{13.81}&\textbf{12.85}&\textbf{12.06}&\textbf{11.17}&\textbf{6.56}&\textbf{1.83}&\textbf{1.15}&\textbf{1.09}\\  
&RANDOM & 17.34&16.255&15.16&15.16&14.08&8.62&5.36&2.09&2.09\\ 
&PPVC & 15.17&15.12&13.80&13.80&13.09&7.33&3.64&1.67&1.67\\ 
&GIAS & 15.12&14.28&13.39&13.39&12.54&7.09&2.33&1.59&1.59 \\
\midrule
\multirow{4}{*}{\makecell{CP Side}}& cRVR &
\textbf{3.26} &\textbf{3.21} &\textbf{3.22} &\textbf{3.17} &\textbf{3.14} &\textbf{2.68} &\textbf{1.43} &\textbf{1.07} &\textbf{1.05}  \\  
&RANDOM & 16.68 &15.83 &14.19 &14.19 &13.57 &8.18 &4.97 &1.95 &1.95 \\ 
&PPVC & 3.86 &3.94 &4.03 &4.03 &4.24 &3.16 &2.86 &1.52 &1.52 \\ 
&GIAS &5.97 &5.82 &5.62 &5.62 &5.41 &4.95 &2.07 &1.34 &1.34 \\  
\bottomrule
\end{tabular}
}
% \vspace{-4mm}
\end{table*}

\begin{table*}[t]
\centering
\caption{\color{black}BCR performance in the UD side and the CP side while varying UDs' cost parameter $\beta$ under the ML-25M Dataset. A lower BCR is preferred in resource consumption.}
\renewcommand{\arraystretch}{1.25}
\resizebox{0.8\linewidth}{!}{
\begin{tabular}{c|c|>{\centering\arraybackslash}p{1cm}|>{\centering\arraybackslash}p{1cm}|>{\centering\arraybackslash}p{1cm}|>{\centering\arraybackslash}p{1cm}|>{\centering\arraybackslash}p{1cm}|>{\centering\arraybackslash}p{1cm}|>{\centering\arraybackslash}p{1cm}|>{\centering\arraybackslash}p{1cm}|>{\centering\arraybackslash}p{1cm}}
\toprule
\multirow{2}{*}{BCR} &\multirow{2}{*}{Methods} & \multicolumn{9}{c}{UDs' Requesting and Caching Cost $\beta$} \\ 
\cline{3-11}
&&0.01  &   0.025 & 0.05 &  0.075 & 0.1 & 0.25 & 0.5  &   0.75& 1.0 \\ 
\midrule
\multirow{4}{*}{UD Side}&cRVR & \textbf{4.73} &\textbf{4.65} &\textbf{4.54} &\textbf{4.48} &\textbf{4.41} &\textbf{4.11} &\textbf{3.69} &\textbf{3.24} &\textbf{2.64}\\  
&RANDOM & 41.58 &40.61 &38.63 &37.62 &36.65 &31.68 &24.76 &19.81 &14.86 \\ 
&PPVC & 27.65 &26.69 &24.69 &23.70 &22.70 &17.73 &10.79 &5.86 &2.16 \\ 
&GIAS & 32.18 &31.42 &29.94 &29.21 &28.45 &23.95 &19.45 &15.65 &12.61 \\
\midrule
\multirow{4}{*}{CP Side}& cRVR & \textbf{3.04} &\textbf{3.01} &\textbf{3.00} &\textbf{2.99} &\textbf{2.97} &\textbf{2.86} &\textbf{2.66} &\textbf{2.38} &\textbf{2.00} \\  
&RANDOM & 19.23 &19.22 &19.18 &19.16 &19.10 &18.67 &17.26 &15.52 &12.98 \\ 
&PPVC &7.71 &7.55 &7.18 &7.00 &6.96 &5.83 &4.19 &3.77 &2.43 \\ 
&GIAS & 7.33 &7.31 &7.24 &7.20 &7.16 &6.90 &6.54 &6.10 &5.64 \\  
\bottomrule
\end{tabular}
}
% \scriptsize
% \begin{tabbing}
% \textbf{Note:} Lower is better ($\downarrow$). Higher is better ($\uparrow$). 
% \end{tabbing}
% \vspace{-2mm}
\label{table:consumption in ML-25M}
% \vspace{-4mm}
\end{table*}

%michael: too small to put in single column. Make it bigger.
%zxz: okay, that will be better.
\begin{table*}[t]
\caption{\color{black}Results of the Cache Churn Rate of cRVR under different parameters settings in the TC and ML-25M datasets.}\label{table:churn rate}
\centering
\renewcommand{\arraystretch}{1.2}
\resizebox{0.58\linewidth}{!}{
\begin{tabular}{c|c|c|c|c|c|c}
\toprule
\multicolumn{2}{c|}{UDs' Privacy Cost $\gamma$}& 0.01  &   0.05  & 0.1  & 0.5 & 1.0 \\ 
\midrule
\multirow{2}{*}{\makecell{Cache Churn Rate \\ ( $\beta_E=0.1, \gamma=0.1$ )}}& TC Dataset & 0.385& 0.139   & 0.055& 0.017& 0.001 \\ 
\cline{2-7}
& ML-25M & 0.068& 0.067   & 0.066& 0.046& 0.019 \\ 
\midrule
\multicolumn{2}{c|}{UDs' Requesting and Caching Cost $\beta$} & 0.01  &   0.05  & 0.1  & 0.5 & 1.0  \\ 
\midrule
\multirow{2}{*}{\makecell{Cache Churn Rate \\ ( $\beta_E=0.1, \gamma=0.1$ )}}& TC Dataset& 0.078& 0.068   & 0.055& 0.004& 0.001 \\ 
\cline{2-7}
& ML-25M & 0.067& 0.066  & 0.065 &0.053 & 0.032 \\ 
\bottomrule
\end{tabular}
}
% \scriptsize
% \begin{tabbing}
% \textbf{Note:} Lower is better ($\downarrow$). Higher is better ($\uparrow$). 
% \end{tabbing}
\vspace{-4mm}
\end{table*}

\subsubsection{\textcolor{black}{\textbf{
%Evaluating 
Caching Performance and Resource Consumption on UDs}}} 
We first conduct two groups of experiments on the TC dataset to evaluate the caching performance of each algorithm on UDs. 
We only consider the probability that each redundant video is watched in the future by UDs, which is measured by CHR. 
In the first group of experiments, we vary $\gamma$ by fixing $\beta=0.1$ and $\beta_E=0.1$, while varying $\beta$ by fixing $\gamma=0.1$ and $\beta_E=0.1$ in the second group to evaluate CHR under different settings. 
The results are shown in Fig.~\ref{Fig:tc_UD}(b) and Fig.~\ref{Fig:tc_UD}(c), from which we can observe that cRVR achieves the highest overall CHR in our experiments and RANDOM is the worst.
%one. 
NR is the benchmark without any redundant requesting and local caching. 
cRVR performs only slightly worse than the state-of-the-art 
%michael: pls double check AIGS is the right one. //ok
%zxz: Thank for the warning. I have ckecked it and it should be the Generating Initial Anonymous Set (GIAS) algorithm. 
% AIGS 
GIAS 
and PPVC algorithms when $\beta$ is very large, as it prioritizes minimizing the requesting and caching costs.
\textcolor{black}{To provide further clarification, we measure additional bandwidth consumption introduced by different request algorithms. Note that, to ensure a fair comparison, we have fixed the upper bound on the number of redundant videos for RANDOM, 
%AIGS, 
GIAS, 
and PPVR to match that of cRVR in each experiment. As shown in Table~\ref{table:consumption in TC}, 
%under 
with 
the same upper limit of requested videos, our algorithm demonstrates a lower BCR under different conditions. This is because, when the caching cost weight parameter is small, cRVR assigns a heavier weight on privacy, carefully selecting videos with higher privacy for caching. However, when the requesting and caching cost weight parameters are higher, cRVR reduces redundant bandwidth consumption to better accommodate resource-constrained devices.
}

To demonstrate the applicability of our proposed algorithm, we also replicate the experiments on the ML-25M dataset, as depicted in Fig.~\ref{Fig:UD_ml25m}(b)-(c) and Table~\ref{table:consumption in ML-25M}. These figures reveal the similar trends that have been presented in Fig.~\ref{Fig:tc_UD}(b)-(c), where the CHR increases with the decrease of $\gamma$ or $\beta$. In the tests conducted on the ML-25M dataset, cRVR demonstrates superior caching performance across various parameter settings. Overall, cRVR exhibits the highest CHR in both datasets, surpassing all other algorithms in most cases. {\color{black}In particular, when users are more concerned about privacy (i.e., a larger $\gamma$), cRVR significantly outperforms the other algorithms. It indicates that the privacy and caching efficiency can be well balanced in our algorithm.  Meanwhile, as shown in Table~\ref{table:consumption in ML-25M}, cRVR consistently achieves the lowest BCR on the UD side across all cost parameter settings, further demonstrating its feasibility.

Lastly, we present the cache churn rate of cRVR under different parameters, $\gamma$ or $\beta$, for the both datasets in Table~\ref{table:churn rate}. 
The cache churn rate is calculated by dividing the number of replaced videos in each round by the total cache capacity and is then averaged across all clients.
As shown in Table~\ref{table:churn rate}, the cRVR algorithm exhibits a low cache churn rate on the UD side, demonstrating that cRVR can achieve effective privacy protection and high user QoE with an acceptable caching resource consumption.
}

%michael: too small the figures. make it bigger. 
%zxz: Okay. I have made it display on two columns
\begin{figure*}[t]
\begin{center}
\hspace{20mm}
%\begin{minipage}[t]{0.235\textwidth}
\begin{minipage}[t]{0.49\textwidth}
\begin{subfigure}{0.7\linewidth}
\includegraphics[width=\linewidth]{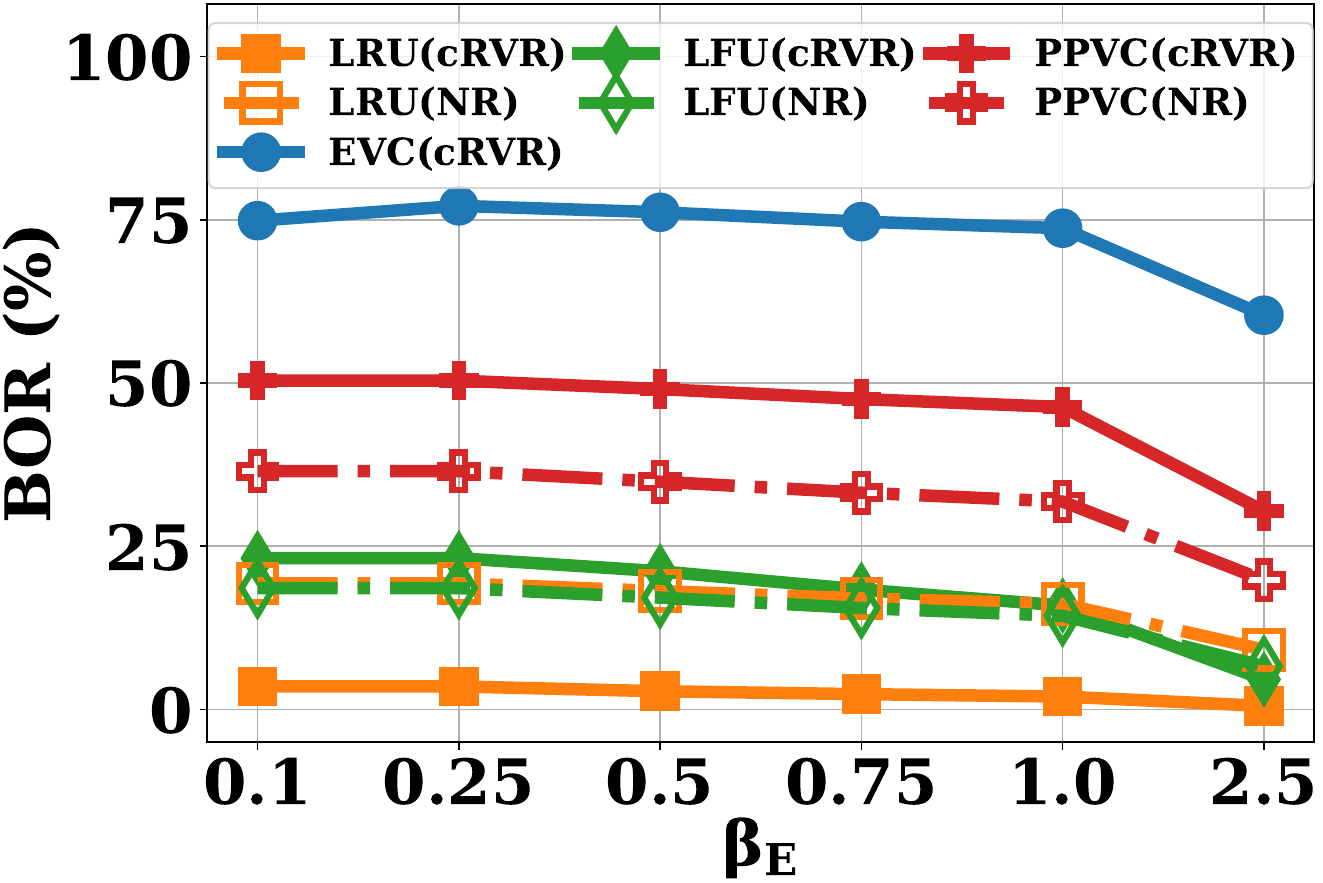}
\caption{TC Dataset.}
\label{Fig:BOR_beta_E}
\end{subfigure}
\end{minipage}
% \hfill
\hspace{-19mm}
%\begin{minipage}[t]{0.235\textwidth}
\begin{minipage}[t]{0.49\textwidth}
\begin{subfigure}{0.7\linewidth}
\includegraphics[width=\textwidth]{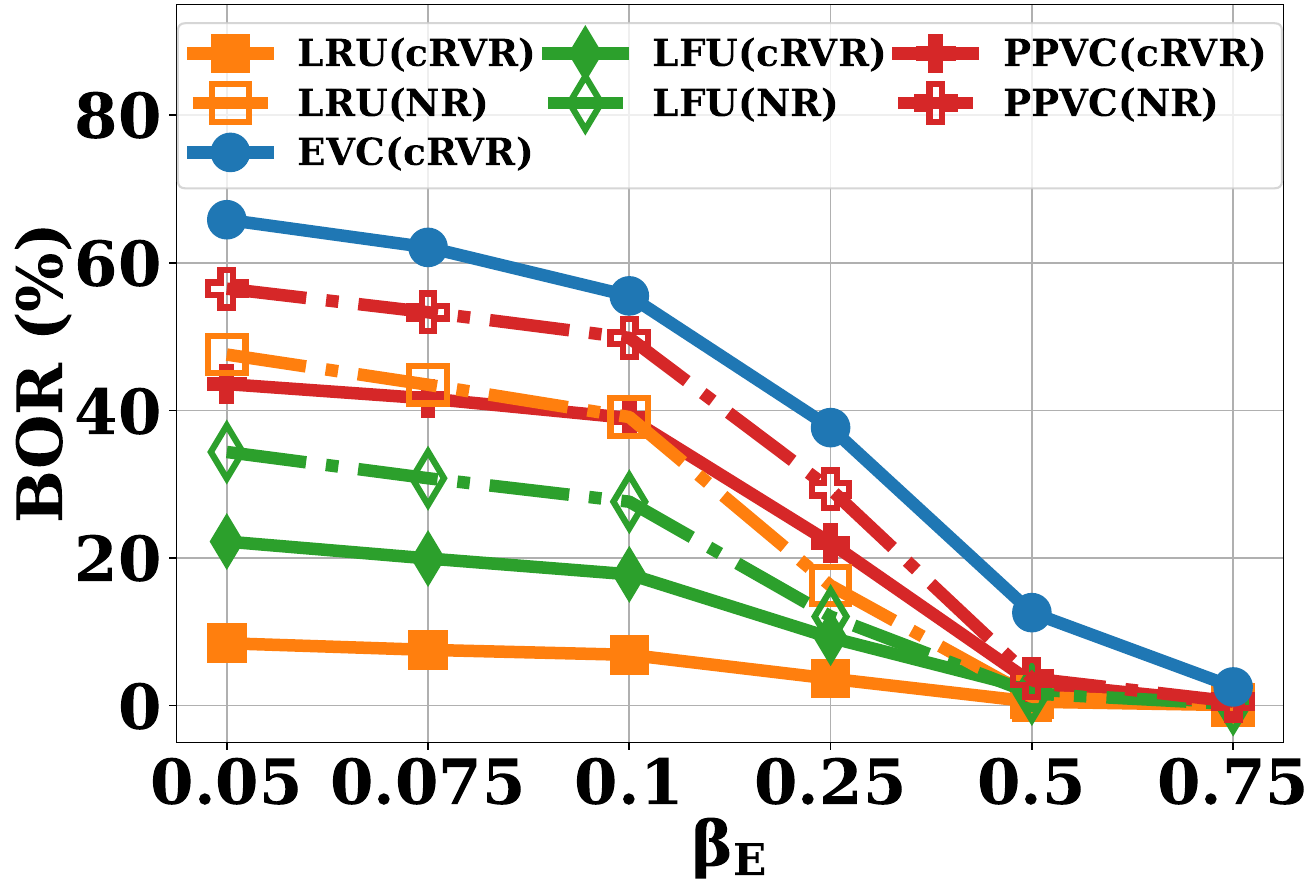}
\caption{ML-25M Dataset.}
\label{Fig:BOR_beta_E_ml}
\end{subfigure}
\end{minipage}
% \hspace{-2mm}
\caption{
%The value of 
EC's BOR when varying the caching cost $\beta_{E}$ with cRVR and NR on UDs.
%, respectively. 
NR is 
%used as 
the benchmark for which the
EC makes caching decisions based on genuine user requests.}
\label{Fig:BOR}
\vspace{-6mm}
\end{center}
\end{figure*}

\vspace{-2mm}
\subsubsection{\textbf{
%Evaluating 
Impacts on 
%the 
CP and EC}} 
{\color{black}
%Furthermore, 
To evaluate the influence of cRVR on the content provider (CP), we compare 
%the 
BCR on the CP side and BOR in EC under varying cost parameters $\beta$ and $\beta_E$ for both %the 
TC and ML-25M datasets, as shown in Tables~\ref{table:consumption in TC}-\ref{table:consumption in ML-25M} and Fig.~\ref{Fig:BOR}, respectively. 

Firstly, to show the impact of the core network (or the CP), we test BCR for all requesting strategies. Here, BCR is defined as the ratio of the actual bandwidth consumption of the request algorithm to the consumption of the NR method, which does not request any additional videos. 
From the tables, it is evident that cRVR demonstrates the lowest BCR values across all cost parameters compared to other strategies like RANDOM, PPVC, and GIAS in both datasets. This indicates that cRVR has the least impacts on the CP, especially under higher cost conditions, %which 
highlighting the effectiveness of our game design in minimizing the influence of redundant requests on the CP. In particular, the influence of cRVR becomes negligible when the cost parameter 
$\beta$ is lower, meaning that UDs can alleviate the impact on the CP by generating redundant requests at a lower frequency. Therefore, 
%using 
with the help of cRVR, UDs can manage redundant requests more efficiently, reducing the load on the core network.
}

We continue to compare the EC performance in terms of BOR by implementing different algorithms.  We implement two algorithms on UDs: cRVR and NR. 
NR is used as the benchmark for which the EC can make caching decisions based on genuine user requests. With this benchmark, we can measure the influence caused by redundant requests. %Meanwhile, 
We implement our EVC algorithm, 
%(ours), 
as well as LRU, LFU and PPVC algorithms on the EC for video caching. We vary $\beta_E$, the weight of caching cost on the EC as the x-axis and fix $\gamma=0.1$, $\beta=0.1$ for EVC. 
Other edge caching baselines are conducted with the same caching capacity as EVC.
The y-axis represents the BOR of each algorithm. 
%michael: black or black??
%zxz: It should be blackt should be black. These changes were made based on the review comments from the previous conference.
The results obtained from two datasets are illustrated in Fig.~\ref{Fig:BOR}(\subref{Fig:BOR_beta_E}) and Fig.~\ref{Fig:BOR}(\subref{Fig:BOR_beta_E_ml}). In the TC dataset, our EVC algorithm attains the highest BOR, exhibiting an average improvement of approximately 27.15\% compared to the second-best algorithm, namely PPVC. Similarly, 
%when employing 
for the ML-25M dataset, EVC surpasses all baselines, showcasing an average 12.37\% improvement in BOR over PPVC.
These experiments validate that cRVR has minimal influence on the CP
through playing the Stackelberg game since most redundant requests can be served by the EC.
It is worth noting that cRVR has negligible impact on popularity-based edge caching algorithms (e.g., LFU and PPVC).
The reason is that cRVR mainly generates requests for popular videos such that viewing patterns of UDs are similar to each other to conceal viewing privacy. 
These popular videos will be cached by the EC with a higher probability resulting in a higher BOR.
Thereby, the influence of cRVR on the core network is insignificant since most redundant requests can be easily served by the EC.

\vspace{-2mm}
\subsubsection{\textbf{Variation Trend of UDs' Utility in Dynamic Games}}
We study the process of how UDs play a game to improve their utilities in Fig.~\ref{Fig:Utility_UD}. Specifically, we fix $\gamma=0.1$, $\beta=0.1$ and $\beta_E=0.1$. To observe the change of each term in the utility, we decompose the utility
%can be decomposed 
into three parts: \textit{Utility-U(L)} 
%representing 
(caching benefits), \textit{Utility-U(D)} 
%representing 
(privacy disclosure), and \textit{Utility-U(C)} 
%representing the 
(caching cost), which correspond to three terms in Eq.~\eqref{Eq:U_u^t}, respectively. 
Each point in Fig.~\ref{Fig:Utility_UD} represents the average utility over all UDs. 
As we can see in Fig.~\ref{Fig:Utility_UD}, all utilities periodically fluctuate, but gradually %get 
become 
stable as UDs play the game for more iterations. 
The results indicate stability and reliability as UDs 
%selfishly 
maximize their own utilities for generating video requests. 

\begin{figure}[!t]
\centering
\includegraphics[width=0.75\linewidth]{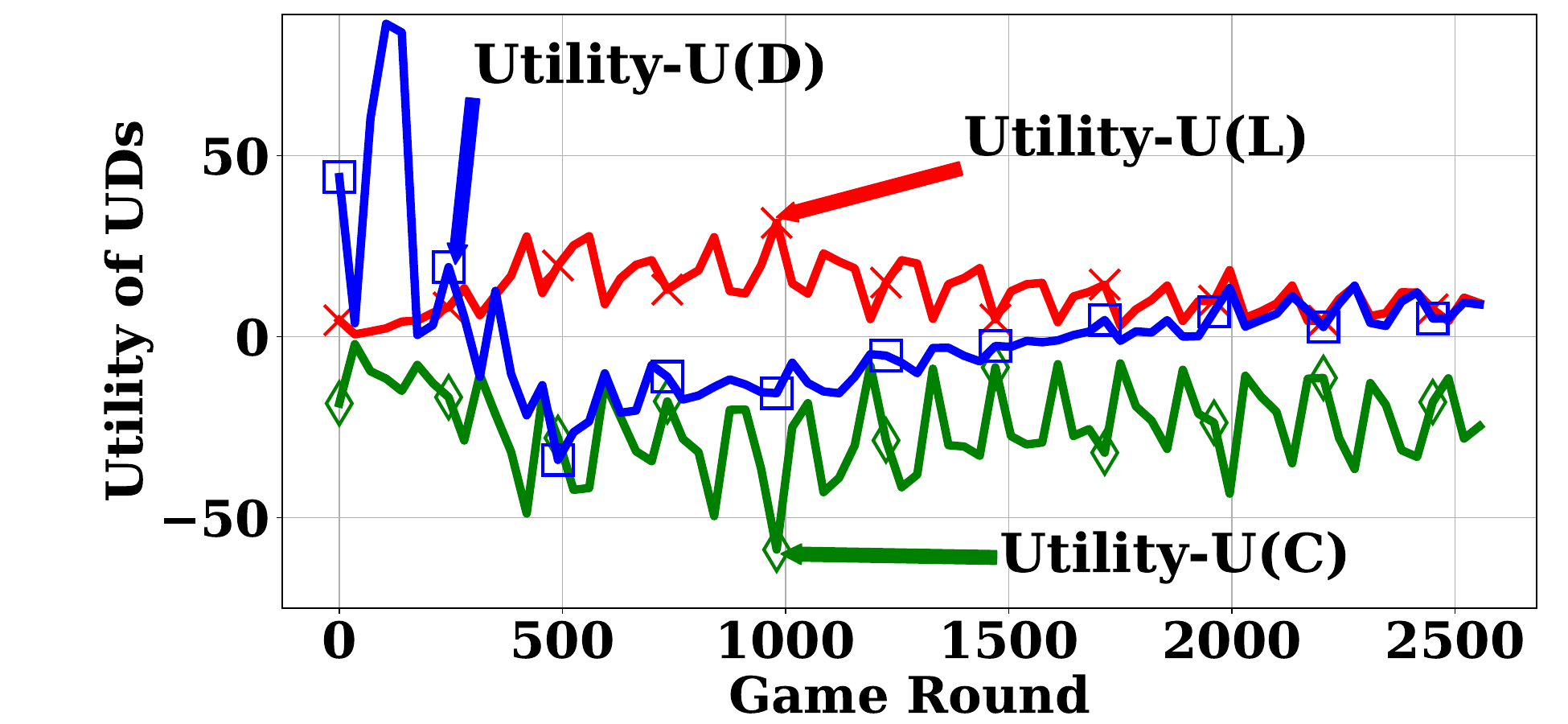}
\caption{The average UDs' utility in each term shown in Eq.~\eqref{Eq:U_u^t} under the default parameter settings (i.e., $\gamma = 0.1$, $\beta = 0.1$, $\beta_{E} = 0.1$) tested on the TC dataset.}
\label{Fig:Utility_UD}
\vspace{-4mm}
\end{figure}

\vspace{-2mm}
\section{Conclusion} \label{sec:conclusion}

Although online video services have been pervasively deployed, there is a lack of effective approaches to preserving request privacy. Since online video systems can automatically capture video request records, we are among the first to particularly investigate how to protect viewers' request privacy by generating redundant video requests.  To mitigate the adverse influence on edge cache, a Stackelberg game is established through which the 
%EC 
%michael: it is a good writing style that avoid using short names in the conclusion.
%zxz: Thank you for the suggestion; I will keep it in mind for my future writing.
edge cache and 
%UDs 
user devices 
can optimize their utilities.
%, respectively. 
The existence and uniqueness of NE (Nash Equilibrium) of the Stackelberg game guarantee the system's stability, which can optimally trade-off caching performance and privacy protection. Our research opens a new door to conceal privacy for users whose records 
%will be 
are 
automatically exposed to the public.
\textcolor{black}{In the next step, we will investigate more advanced algorithms for privacy-enhancing edge caching and deploy them in real online video systems. Additionally, considering correlations between videos and detailed user interactions, such as fast forward or rewind, and pause or play, during video playback as a potential research direction will further enhance privacy protection. Obfuscating these interactions from the user side could further reduce the risk of leaking user preferences to content providers.}

\vspace{-2mm}
\bibliographystyle{IEEEtran}
\bibliography{IEEEabrv,reference}

%\ifCLASSOPTIONcompsoc
% \section*{Acknowledgments}
%\else
%\fi
%\ifCLASSOPTIONcompsoc

\vspace{-13mm}
\begin{IEEEbiography}
[{\includegraphics[width=1in,height=1.25in,clip,keepaspectratio]{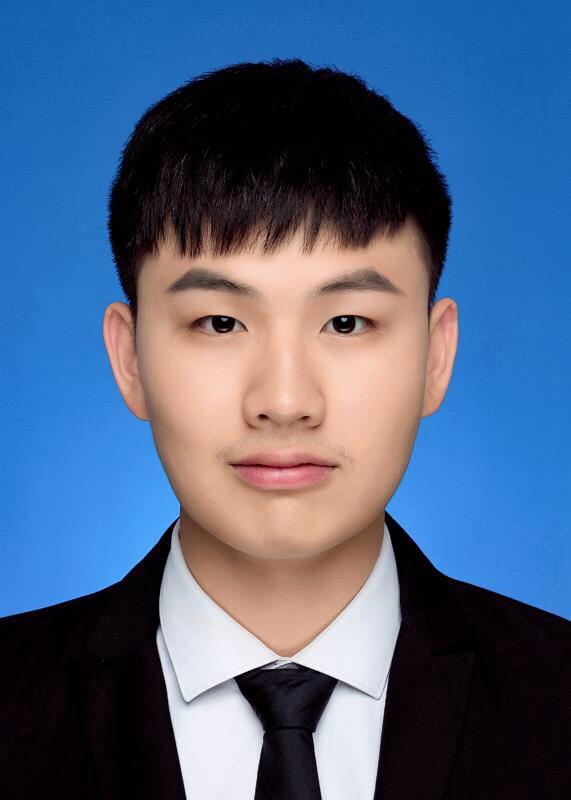}}]
{Xianzhi Zhang} received his B.S. degree from Nanchang University (NCU), Nanchang, China, in 2019 and an M.S. degree from the School of Computer Science and Engineering, Sun Yat-sen University (SYSU), Guangzhou, China, in 2022.  He is currently pursuing a Ph.D. degree in the School of Computer Science and Engineering at Sun Yat-sen University, Guangzhou, China. He is also 
%working as 
a visiting PhD student at the School of Computing,  Macquarie University, Sydney, Australia. Xianzhi's current research interests include video caching, privacy protection, federated learning, and edge computing. His researches have been published at IEEE TPDS and TSC, and won the Best Paper Award at PDCAT 2021. 
\vspace{-10mm}
\end{IEEEbiography}

\begin{IEEEbiography}[{\includegraphics[width=1in,height=1.25in,clip,keepaspectratio]{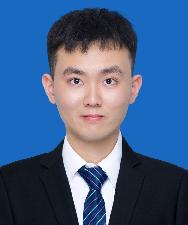}}]
{Linchang Xiao} received his B.S. Degree from the School of Computer Science and Engineering, Sun Yat-sen University (SYSU), Guangzhou, China, in 2022. He is currently working toward the M.S. degree at the School of Computer Science and Engineering, Sun Yat-sen University (SYSU), Guangzhou, China. His research interests include cloud \& edge computing, content caching, and multimedia communication. 
\vspace{-10mm}
\end{IEEEbiography}

\begin{IEEEbiography}[{\includegraphics[width=1in,height=1.25in,clip,keepaspectratio]{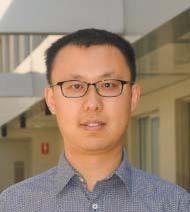}}]
{Yipeng Zhou} is a senior lecturer 
%in computer science 
in School of Computing at Macquarie University, and the recipient of ARC DECRA in 2018. From Aug. 2016 to Feb. 2018, he was a research fellow with Institute for Telecommunications Research (ITR) of University of South Australia. From 2013.9-2016.9, He was a lecturer in College of Computer Science and Software Engineering, Shenzhen University. He was a Postdoctoral Fellow with Institute of Network Coding (INC) of The Chinese University of Hong Kong (CUHK) from Aug. 2012 to Aug. 2013. He won his PhD degree and Mphil degree from Information Engineering (IE) Department of CUHK respectively. He got Bachelor degree in Computer Science from University of Science and Technology of China (USTC). His research interests lie in federated learning, privacy protection and caching algorithm design in networks. He has published more than 80 papers including IEEE INFOCOM, ICNP, IWQoS, IEEE ToN, JSAC, TPDS, TMC, and TMM.
%, etc.
\vspace{-10mm}
\end{IEEEbiography}

\begin{IEEEbiography}[{\includegraphics[width=1in,height=1.25in,clip,keepaspectratio]{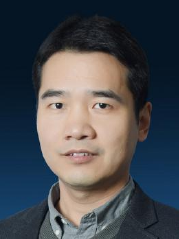}}]
{Di Wu}(M'06-SM'17) received the B.S. degree from the University of Science and Technology of China, Hefei, China, in 2000, the M.S. degree from the Institute of Computing Technology, Chinese Academy of Sciences, Beijing, China, in 2003, and the Ph.D. degree in computer science and engineering from the Chinese University of Hong Kong, Hong Kong, in 2007. He was a Post-Doctoral Researcher with the Department of Computer Science and Engineering, Polytechnic Institute of New York University, Brooklyn, NY, USA, from 2007 to 2009, advised by Prof. K. W. Ross. Dr. Wu is currently a Professor and the Associate Dean of the School of Computer Science and Engineering, Sun Yat-sen University, Guangzhou, China. He was the recipient of the IEEE INFOCOM 2009 Best Paper Award and IEEE Jack Neubauer Memorial Award.
%, and etc. 
His research interests include edge/cloud computing, multimedia communication, Internet measurement, and network security.
\vspace{-10mm}
\end{IEEEbiography}

%\vspace{-5mm}
\begin{IEEEbiography}[{\includegraphics[width=1in,height=1.25in,clip,keepaspectratio]{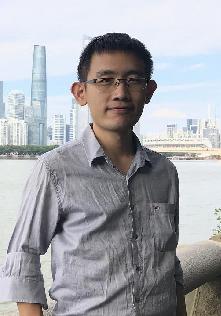}}]
{Miao Hu} (S'13-M'17) is currently an Associate Professor in the School of Computer Science and Engineering, Sun Yat-Sen University, Guangzhou, China. He received the B.S. degree and the Ph.D. degree in communication engineering from Beijing Jiaotong University, Beijing, China, in 2011 and 2017, respectively. From Sept. 2014 to Sept. 2015, he was a Visiting Scholar with the Pennsylvania State University, PA, USA. His research interests include edge/cloud computing, multimedia communication and software defined networks.
\vspace{-10mm}
\end{IEEEbiography}

%\vspace{-5mm}
\begin{IEEEbiography}[{\includegraphics[width=1in,height=1.25in,clip,keepaspectratio]{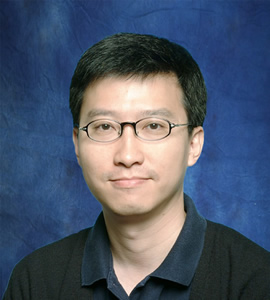}}]
{John C.S. Lui} is currently the Choh-Ming Li Chair Professor in the Department of Computer Science \& Engineering (CSE) at The Chinese University of Hong Kong (CUHK). He received his Ph.D. in Computer Science from UCLA. After his graduation, he joined the IBM Lab and participated in research and development projects on file systems and parallel I/O architectures. He later joined the CSE Department at CUHK. 
%He has been a visiting professor in computer science departments at UCLA, Columbia University, University of Maryland at College Park, Purdue University, University of Massachusetts at Amherst and Universit degli Studi di Torino in Italy. 
Currently, he is leading 
%a group of students and post-docs in 
the Advanced Networking and System Research Laboratory (ANSRLab). His research interests include online learning algorithms and applications (e.g., multi-armed bandits, reinforcement learning), quantum Internet, machine learning on network sciences and networking systems, large-scale data analytics, network/system security, network economics, large scale storage systems and performance evaluation theory. 
%John is an active consultant to industry, believing that it is an effective way to do technology transfer and a wonderful way to learn about real and relevant research problems. 
John is 
%currently the 
a senior editor in the IEEE/ACM Transactions on Networking, and has been serving in the editorial board of several leading journals including ACM Transactions on Modeling and Performance Evaluation of Computing Systems, IEEE Transactions on Network Science \& Engineering, IEEE Transactions on Mobile Computing, IEEE Transactions on Computers, and IEEE Transactions on Parallel and Distributed Systems.
%Journal of Performance Evaluation, Journal of Network Science and International Journal of Network Security.
\vspace{-10mm}
\end{IEEEbiography}

%\vspace{-5mm}
\begin{IEEEbiography}[{\includegraphics[width=1in,height=1.25in,clip,keepaspectratio]{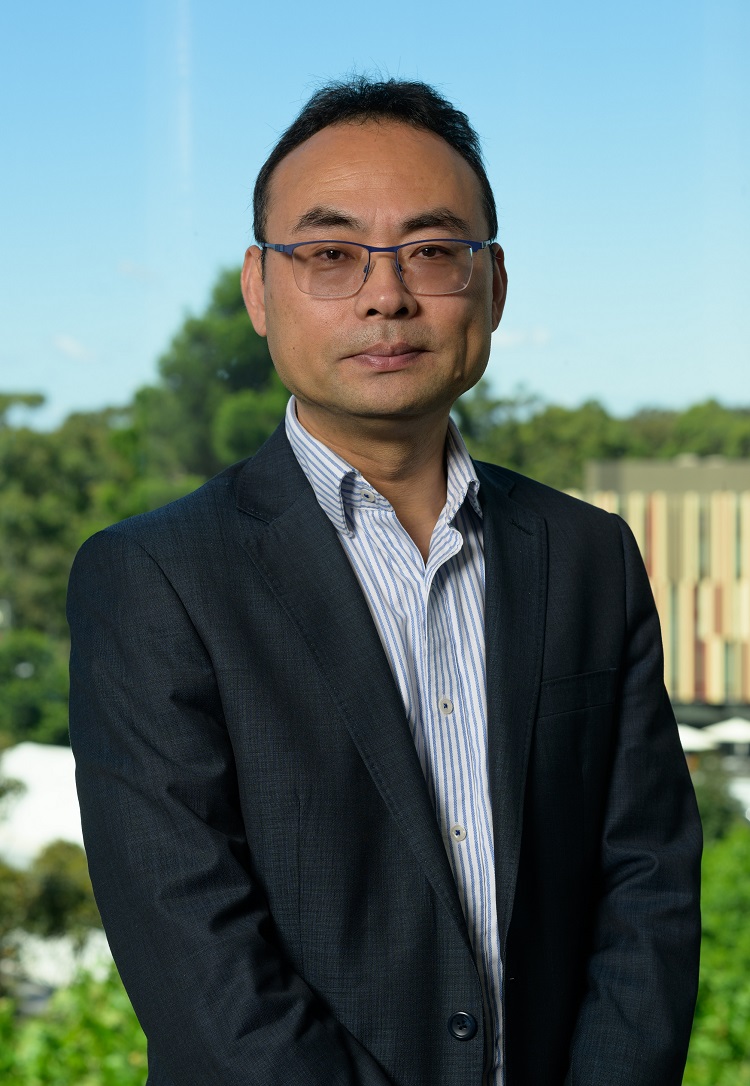}}]
{Quan Z. Sheng} is a Distinguished Professor and Head of School of Computing at Macquarie University. Before moving to Macquarie, Michael spent 10 years at School of Computer Science, the University of Adelaide, serving in senior leadership roles such as Interim Head of School and Deputy Head of School. Michael holds a PhD degree in computer science from the University of New South Wales (UNSW) and did his post-doc as a research scientist at CSIRO ICT Centre. 
%From 1999 to 2001, Sheng also worked at UNSW as a visiting research fellow. Prior to that, he spent 6 years as a senior software engineer in industries. 
%
Prof. Michael Sheng's research interests include Web of Things, Internet of Things, Big Data Analytics, Web Science, Service-oriented Computing, Pervasive Computing, and Sensor Networks. He is ranked by Microsoft Academic as one of the Most Impactful Authors in Services Computing (Top 5 all time), and is ranked by ScholarGPS as Top 5 Lifetime in Web Information Systems. 
%and in Web of Things (ranked Top 20 all time). 
He is the recipient of the AMiner Most Influential Scholar Award on IoT (2019), ARC Future Fellowship (2014), Chris Wallace Award for Outstanding Research Contribution (2012), and Microsoft Research Fellowship (2003). Prof Sheng is Vice Chair of the Executive Committee of the IEEE Technical Community on Services Computing (IEEE TCSVC) 
%, the Associate Director (Smart Technologies) of Macquarie's Smart Green Cities Research Centre, 
and a member of the ACS Technical Advisory Board on IoT.
\end{IEEEbiography}

\clearpage
\begin{appendices}

\begin{figure*}[tbp]
\includegraphics[width=\linewidth]{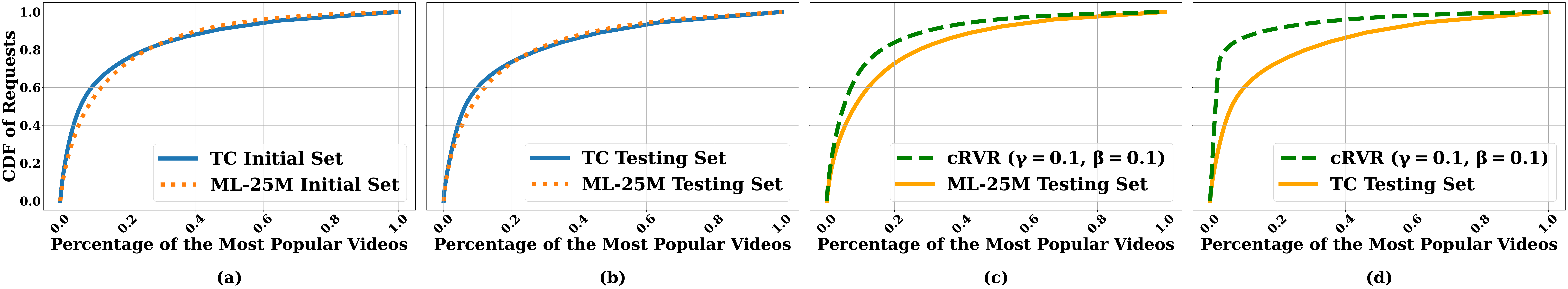}
\vspace{-4mm}
\caption{CDF of Video Request Popularity Distribution, where (a) and (b) illustrate the CDF of video requests based on the most popular videos across two datasets while (c) and (d) show the impact of redundant requests introduced by the cRVR algorithm on the popularity distribution in both the TC and ML-25M testing sets.}
%\vspace{-4mm}
\label{Fig:cdf}
\end{figure*}
    
% \setcounter{secnumdepth}{0}
% \appendix
\section{Proof of Theorem \ref{theorem:eds exist}}\label{appendix:eds exist}
\begin{proof}
When $\bm{x}_{u,i}^t=1$, the optimal strategy $a_{u,i}^{t*}=1$ to meet the viewing demand.
%Otherwise, when $x_{u,i}^t=0$, we have $a_{u,i}^{t*}=\lfloor y_{u,i}^{t*} + 1/2 \rfloor,\, y_{u,i}^{t*}\in[0,1]$.
Without loss of gentility, we use $y_{u,i}^t$  to analyze instead of $a_{u,i}^t$.
Recall that $N_i^t = n_i^t - m_i^t- 2\Delta n_i^t$.
The first order differential of the UD’s utility with respect to $y_{u,i}^{t}$ is 
\begin{equation}
\frac{\partial \mathrm{U}_{u}^{t}}{\partial y_{u,i}^{t}} 
=d_{u,i}^{t}\hat{p}_{i}^{t}c_{i}(1-e_i^{t})-\beta\,c_i\,\epsilon_u+\frac{\gamma\left(1-r_{u,i}^t\right)N_i^t}{y_{u,i}^{t}\,N_i^t-(n_i^{t}-\Delta n_i^{t})}.
\end{equation}

The second order differential of the UD’s utility with
respect to $y_{u,i}^{t}$ is \begin{equation}
\frac{\partial^2 \mathrm{U}_{u}^{t}}{{\partial y_{u,i}^{t}}^2} 
= -\gamma\left(1-r_{u,i}^t\right)\left[\frac{N_i^t}{y_{u,i}^{t}\,N_i^t-(n_i^{t}-\Delta n_i^{t})}\right]^2.
\end{equation}

We consider two kinds of UDs.
%, respectively. 
First, we consider the UD $u$ who has not requested content $i$ before period $t$, i.e., $r_{u,i}^t=0$.
Correspondingly, the second order differential of the UE’s utility 
%  $\frac{\partial^2 \ U_{u}^{t}}{{\partial \ a_{u,i}^{t}}^2}<0$
$\frac{\partial^2 \mathrm{U}_{u}^{t}}{{\partial y_{u}^{t}}^2}$ is less than $0$, %and the utility function $\mathrm{U}_{u}^t$ of UD is a strictly concave function, w
where there exists a maximum utility. Let $\frac{\partial \mathrm{U}_{u}^t}{\partial \ y_{u,i}^t}=0$, the optimal strategy can be obtained by a best respond function to other UDs in Eq.\eqref{Eq:response function case a} with the constrain $y_{u,i}^{t}\in[0,1]$. 
%With the constrain $y_{u,i}^{t}\in[0,1]$, the best cRVR strategy of UD $u$ who has not request record on content $i$ before time $t$ can be summarized in Eqs.\eqref{Eq:response function case a} and \eqref{Eq:response function B} with $r_{u,i}^t=0$.
Next, we consider a UD $u$ in status $r_{u,i}^t=1$, and the first order differential of the UE's utility becomes 
\begin{equation}
    \frac{\partial \mathrm{U}_{u}^{t}}{\partial \ y_{u,i}^{t}} 
=d_{u,i}^{t}\,\hat{p}_{i}^{t}\,c_{i}\,(1-e_i^{t})-\beta\,c_i\,\epsilon_u.
\end{equation}
Here, the value of the first order differential is a constant when $e_i^{t}$ is given and the second order differential of the UE's utility is $0$. 
Therefore, the UD’s utility is a monotone function. The optimal strategy $y_{u,i}^{t*}$, shown in Eq.\eqref{Eq:response function case b}, is obtained at $y_{u,i}^{t*} = 0$ or $y_{u,i}^{t*} = 1$ depending on the value difference between $\mathrm{U}_{u,i}^t(e_i^t,1,\bm{a}_{-u,i}^t)$ and $\mathrm{U}_{u,i}^t(e_i^t,0,\bm{a}_{-u,i}^t)$\footnote{When $\frac{\partial\, \mathrm{U}_{u}^t}{\partial a_{u,i}^t}=0$, all strategies have the same utility and we define the $y_{u,i}^{t*}= 0$ is the only optimal strategy.}.
%As such, if $e_i^t<1-\frac{\beta\,\epsilon_u}{d_{u, i}^t\,\hat{p}_{u, i}^t}$, where $d_{u,i}^{t}\neq0,\hat{p}_{i}^{t}\neq0$, the UE’s utility is a monotone increasing function in interval $[0,1]$ and a maximum utility can be obtained at the max end of $y_{u,i}^t$ interval, i.e., $y_{u,i}^{t*}= 1$. Otherwise, the UE’s utility is a decreasing function in interval $[0,1]$, where a maximum utility can be obtained at $y_{u,i}^{t*}= 0$.\footnote{When $\frac{\partial\, \mathrm{U}_{u}^t}{\partial a_{u,i}^t}=0$, all strategies have the same utility and define the $y_{u,i}^{t*}= 0$ is the only optimal strategy.} 
%Specially, as for the UD $u$ with the status $d_{u,i}^{t}=0$ or $\hat{p}_{i}^{t}=0$, we have $\frac{\partial \ \mathrm{U}_{u}^{t}}{\partial \ a_{u,i}^{t}}=-\beta\,c_i\,\epsilon_u<0$ and it is better for UD $u$ to take the optimal strategy $y_{u,i}^{t*}=0$.
In a word, the best cRVR strategy of UD $u$ who has requested on content $i$ before time $t$ can be summarized in Eq.\eqref{Eq:response function case} with $r_{u,i}^t=1$.
%In conclusion, 
%The 
Theorem \ref{theorem:eds exist} is 
therefore 
proved.
\end{proof}

\section{Supplementary Proof  of Theorem \ref{theorem:eds unique}}\label{appendix: eds unique}
\begin{proof}
In this subsection, we show that $y_{u,i}^{t*}(\Delta n_i^{t*})$ in Eq.~\eqref{Eq:response function case a} satisfies the three properties of a standard function.

\vspace{2mm}
\noindent (1) \textit{Positivity}: The best response of the $y_{u,i}^{t*}$ is always positive, as $y_{u,i}^{t*}(\Delta n_i^{t*})\in[0,1]$.

\vspace{2mm}
\noindent (2) \textit{Monotonicity}: the first order differential of the function $\Omega(\Delta n_i^{t*})$ is 
\begin{equation}
    \frac{\partial\, \Omega(\Delta n_i^{t*}) }{\partial \Delta n_i^{t*}}=\frac{m_i^t+n_i^t}{\left(n_i^t-m_i^t-2\Delta n_i^{t*}\right)^{2}}>0.
\end{equation}
Therefore, for any $\Delta n_i^{t*}\geq{\Delta n_i^{t*}}'$, we have $\Omega(\Delta n_i^{t*})\geq{\Omega(\Delta n_i^{t*})}'$. Recall that $\Delta n_i^{t*} = \sum_{\forall u': r_{u',i}^t =0}a_{u',i}^{t*}$. 
When $\bm{a}_{-u,i}^{t*}\succeq{\bm{a}_{-u,i}^{t*}}'$, we have $\Delta n_i^{t*}\geq\Delta {n_i^{t*}}'$.
Therefore, the response function $\Omega(\Delta n_i^{t*})$ is monotonic to any $\bm{a}_{-u,i}^{t*}$.
Considering $y_{u,i}^{t*}(\Delta n_i^{t*}) = \max\{0,\min\{\Omega(\Delta n_i^{t*}),1\}\}$ in Eq.~\eqref{Eq:response function case a}, $y_{u,i}^{t*}(\Delta n_i^{t*})$ is monotonic to any $\bm{a}_{-u,i}^{t*}$.

\vspace{2mm}
\noindent (3) \textit{Scalability}: For all $k\geq1$, we need to prove that $k\cdot y_{u,i}^{t*}(\Delta n_i^{t*})\geq y_{u,i}^{t*}(k\cdot\Delta n_i^{t*})$.  
% Recall that $e_i^t\in[0,1]$, $N_i^{t*} = n_i^t - m_i^t- 2\Delta n_i^{t*}>0$, and $n_i^t>\Delta n_i^{t*}$. 
% For simplicity, we also let \begin{equation}
%     g(k)=-\frac{ n_i^t-k\,\Delta n_i^{t*}}{n_i^t-m_i^t-2\,k\,\Delta n_i^{t*}},
% \end{equation} 
% where 
% \begin{equation}
% g'(k)=\frac{\Delta n_i^{t*}(n_i^t+m_i^t)}{(n_i^t-m_i^t-2\,\Delta n_i^{t*})^2}\geq 0,\end{equation} 
% and we have:
% \begin{equation}
% \begin{aligned}
% &k\Omega(\Delta n_i^{t*})-\Omega(k\,\Delta n_i^{t*})\\
% =&\frac{(k-1)\,\gamma}{\beta\,c_i\,\epsilon_u-c_i\,d_{u,i}^t\,\hat{p}_{i}^t(1-e_i^t)}+\frac{k\,(n_i^t-\Delta n_i^{t*})}{N_i^{t*}}
% +g(k)\\
% \geq&\frac{(k-1)\,\gamma}{\beta\,c_i\,\epsilon_u}+\frac{k\,(n_i^t-\Delta n_i^{t*})}{N_i^{t*}}+g(k)\\
% \geq&\frac{(k-1)\,\gamma}{\beta\,c_i\,\epsilon_u}+\frac{k\,(n_i^t-\Delta n_i^{t*})}{N_i^{t*}}+g(1)\\
% =&\frac{(k-1)\,\gamma}{\beta\,c_i\,\epsilon_u}+\frac{(k-1)\,(n_i^t-\Delta n_i^{t*})}{N_i^{t*}}\\
% =&(k-1)\,\left(\frac{\gamma}{\beta\,c_i\,\epsilon_u}+\frac{n_i^t-\Delta n_i^{t*}}{N_i^{t*}}\right)\geq 0,
% \end{aligned}
% \end{equation}
% which indicates that $\Omega(\Delta n_i^{t*})$ is scalable.
For simplicity, 
%we 
let \begin{equation}
    g(k)=\frac{ n_i^t-\Delta n_i^{t*}}{n_i^t-m_i^t-2\,\Delta n_i^{t*}}-\frac{ n_i^t-k\,\Delta n_i^{t*}}{n_i^t-m_i^t-2\,k\,\Delta n_i^{t*}},
\end{equation} 
where 
\begin{equation}
g'(k)=\frac{\Delta n_i^{t*}(n_i^t+m_i^t)}{(n_i^t-m_i^t-2\,k\,\Delta n_i^{t*})^2}\geq 0,\end{equation}
Considering $y_{u,i}^{t*}(\Delta n_i^{t*})= \max\{0,\min\{\Omega(\Delta n_i^{t*}),1\}\}\in[0,1]$ in Eq.~\eqref{Eq:response function case a}, we first prove $k\cdot \Omega_{u,i}(\Delta n_i^{t*})- \Omega_{u,i}(k\cdot\Delta n_i^{t*})\geq 0$ as follows.
\begin{equation}
\begin{aligned}
&k\Omega(\Delta n_i^{t*})-\Omega(k\,\Delta n_i^{t*})\\
=&(k-1)(\frac{\,\gamma}{\beta\,c_i\,\epsilon_u-c_i\,d_{u,i}^t\,\hat{p}_{i}^t(1-e_i^t)}+\frac{n_i^t-\Delta n_i^{t*}}{N_i^{t*}})+g(k)\\
\geq&(k-1)y_{u,i}^{t*}(\Delta n_i^{t*})+g(k)
\geq g(1)=0
\end{aligned}
\end{equation}
% \begin{equation}
% \begin{aligned}
% &k\Omega(\Delta n_i^{t*})-\Omega(k\,\Delta n_i^{t*})\\
% =&\frac{(k-1)\,\gamma}{\beta\,c_i\,\epsilon_u-c_i\,d_{u,i}^t\,\hat{p}_{i}^t(1-e_i^t)}+\frac{k\,(n_i^t-\Delta n_i^{t*})}{N_i^{t*}}
% +g(k)\\
% \geq&\frac{(k-1)\,\gamma}{\beta\,c_i\,\epsilon_u}+\frac{k\,(n_i^t-\Delta n_i^{t*})}{N_i^{t*}}+g(k)\\
% \geq&\frac{(k-1)\,\gamma}{\beta\,c_i\,\epsilon_u}+\frac{k\,(n_i^t-\Delta n_i^{t*})}{N_i^{t*}}+g(1)\\
% =&\frac{(k-1)\,\gamma}{\beta\,c_i\,\epsilon_u}+\frac{(k-1)\,(n_i^t-\Delta n_i^{t*})}{N_i^{t*}}\\
% =&(k-1)\,\left(\frac{\gamma}{\beta\,c_i\,\epsilon_u}+\frac{n_i^t-\Delta n_i^{t*}}{N_i^{t*}}\right)\geq 0,
% \end{aligned}
% \end{equation}
Therefore, $y_{u,i}^{t*}(\Delta n_i^{t*})$ is scalable to any $\bm{a}_{-u,i}^{t*}$ as well.
%To conclude, 
We prove that the response function is standard. 

\end{proof}

\section{Proof of Theorem \ref{theorem:EC}}\label{appendix of theorem EC}
\begin{proof}
The first-order differential of the EC’s utility with respect to $e_i$ is
\begin{equation}
\frac{\partial \mathrm{U}_{E}}{\partial e_i^{t}} 
= \frac{c_i\,\sum_{\forall u} a_{u,i}^{t*}}{1+e_i^{t}\,c_i} - \beta_E\,c_i\, \epsilon_E.
\end{equation}

Moreover, the second order differential of the EC’s utility with respect to $e_i^t$ is:
\begin{equation}
\frac{\partial^2 \mathrm{U}_{E}}{{\partial e_i^t}^2} 
= -\left(\frac{{c_i}}{1+e_i^t\,c_i}\right)^2 \,\sum_{\forall u} a_{u,i}^{t*}.
\end{equation}

We consider two kinds of situations, the first of which is $\sum_{u}a_{u,i}^{t*}>0$. 
Here, the second order differential of edge caching device utility $\frac{\partial^2 \ \mathrm{U}_{E}}{{\partial \ e_i^t}^2}$ is smaller than zero so that the utility function $\mathrm{U}_{E}$ is a strictly concave function, where there exists a maximum value in $[0,1]$ with the unique optimal strategy $e_i^{t*}$. The optimal strategy $e_i^{t*}$ can obtain directly by $\frac{\partial \mathrm{U}_{E}}{\partial  e_i^{t}} =0$.
As for $\sum_{u}a_{u,i}^t=0$, we have the first differential of utility function as \begin{equation}
    \frac{\partial\mathrm{U}_{E}}{\partial e_i^t} = -\beta_E\, \epsilon_E\,c_i<0,
\end{equation}
and the utility function is a strictly decreasing function with the unique best strategy $e_i^{t*}=0$.
Note that, the above analysis can be applied to any content $i$. 
%To sum up, 
Theorem \ref{theorem:EC} is therefore proved.

\end{proof}

\begin{table}[!ht]
\caption{Datasets Summary}
\centering
\renewcommand{\arraystretch}{1.2}
\resizebox{1\linewidth}{!}{
\begin{tabular}{c|c|c|c}
\toprule
Dataset & User Count & Content Count & Density \\ 
\midrule
TC & 1,000 & 20,999 & 1.0652\% \\ 
ML-25M & 1,000 & 10,002 & 1.5799\% \\ 
\bottomrule
\end{tabular}
}
% \vspace{-2mm}
\label{table:dataset_summary}
\end{table}

%michael: and what? in title; if they are different, we can have two sections. //OK
%zxz: I suggest using "Detailed Information on Datasets" as the title for this appendix section.
\section{Detailed Information of Datasets}\label{appendix: dataset}
{\color{black}
In this section, we provide a detailed overview of the two datasets used in our experiments: TC and ML-25M. 
Both datasets represent distinct real-world video streaming scenarios, featuring varying levels of content popularity and request distributions. 
Table~\ref{table:dataset_summary} provides a summary of the two datasets (TC and ML-25M) used in our experiments. It includes key details such as the number of users, the number of video contents, and the data density. As shown in the table, the TC dataset contains 1,000 users and 20,999 videos, with a density of 1.0652\%, while the ML-25M dataset also contains 1,000 users but has fewer video contents (10,002) and a higher density of 1.5799\%. 

Fig.~\ref{Fig:cdf} (a) and (b) display the CDF (cumulative distribution function) of users' requests for both the initial and testing sets of each dataset. From these figures, we observe that the TC dataset exhibits a relatively concentrated popularity distribution, allowing the cRVR algorithm to select more popular videos for redundant requests, thereby enhancing privacy protection. Additionally, by comparing Fig.~\ref{Fig:cdf} (a)-(b), we can see that the popularity distribution across the datasets remains relatively stable.
Fig.~\ref{Fig:cdf} (c)-(d) provide a comparison between the original request distribution and the distribution after applying the cRVR algorithm. The introduction of redundant requests by cRVR aims to obscure users’ true preferences while maintaining efficient caching performance. As demonstrated in these figures, cRVR alters the request patterns, effectively concentrating the requests and thus reducing the amount of exposed user information. This helps protect user privacy while also maintaining caching efficiency.
}

\end{appendices}

\end{document}